\documentclass[
 reprint,
 amsmath,amssymb,
 aps,
]{revtex4-2}

\usepackage{graphicx}
\usepackage{dcolumn}
\usepackage{bm}

\usepackage{amsmath, amsthm, amssymb, physics, mathtools, dsfont}
\usepackage{bbm}
\usepackage{xcolor}
\usepackage{graphicx}
\usepackage{float}
\usepackage[a4paper, left=30mm, right=30mm, top=30mm, bottom=30mm]{geometry}
\usepackage{hyperref}
\hypersetup{pdfborder=0 0 0, colorlinks, allcolors=black}
\usepackage{graphicx}
\usepackage{natbib}
\usepackage[inline]{enumitem}
\usepackage[capitalise]{cleveref}	
\usepackage{leftindex}

\newcommand{\HH}{\mathcal{H}}

\newcommand{\kb}[2]{\ket{#1}\hspace{-1mm}\bra{#2}}

\newcommand{\CC}{\mathbb{C}}
\newcommand{\UU}{\mathcal{U}}
\newcommand{\RR}{\mathbb{R}}
\newcommand{\ZZ}{\mathbb{Z}}
\newcommand{\NN}{\mathbb{N}}
\newcommand{\EE}{\mathcal{E}}

\newcommand{\ID}{\mathds{1}}

\newcommand{\UO}{\operatorname{U}(1)}
\newcommand{\SUT}{\operatorname{SU}(2)}
\newcommand{\OTH}{\operatorname{O}(3)}

\newcommand{\jmax}[1]{\bar J_{#1}}

\newtheorem{prop}{Proposition}

\newtheorem{cor}{Corollary}
\newtheorem{defi}{Definition}
\newtheorem{ex}{Example}

\newenvironment{equation-aligned}{\begin{equation}\begin{aligned}}{\end{aligned}\end{equation}}

\newsavebox{\mstrut}
\newcommand{\bbra}[1]{
    \sbox{\mstrut}{\(#1\)}
    \mathinner{\left\langle\kern-0.5\ht\mstrut\left\langle{#1}\right|\mkern-2mu\right|}
}
\newcommand{\kett}[1]{
    \sbox{\mstrut}{\(#1\)}
    \mathinner{\left|\mkern-2mu\left|{#1}\right\rangle\kern-0.5\ht\mstrut\right\rangle}
}

\usepackage{centernot}
\usepackage{stmaryrd}
\usepackage{mathtools}

\usepackage[normalem]{ulem}

\makeatletter
\newcommand{\xMapsto}[2][]{\ext@arrow 0599{\Mapstofill@}{#1}{#2}}
\def\Mapstofill@{\arrowfill@{\Mapstochar\Relbar}\Relbar\Rightarrow}
\makeatother
\makeatletter
\newlength{\negph@wd}
\DeclareRobustCommand{\negphantom}[1]{
  \ifmmode
    \mathpalette\negph@math{#1}
  \else
    \negph@do{#1}
  \fi
}
\newcommand{\negph@math}[2]{\negph@do{$\m@th#1#2$}}
\newcommand{\negph@do}[1]{
  \settowidth{\negph@wd}{#1}
  \hspace*{-\negph@wd}
}
\makeatother

\expandafter\def\expandafter\normalsize\expandafter{%
    \normalsize%
    \setlength\abovedisplayskip{5pt}%
    \setlength\belowdisplayskip{5pt}%
    \setlength\abovedisplayshortskip{-2pt}%
    \setlength\belowdisplayshortskip{4pt}%
}
\usepackage{titlesec}

\titlespacing*{\section}
{0pt}{3.5ex}{3.3ex }
\titlespacing*{\subsection}
{0pt}{2.5ex }{2.3ex }

\begin{document}

\title{Rotational covariance restricts available quantum states}

\author{Fynn Otto}
\author{Konrad Szyma\'nski}
\affiliation{
Naturwissenschaftlich-Technische Fakultät, Universität Siegen, Walter-Flex-Straße 3, 57068 Siegen, Germany
}

\date{\today}

\begin{abstract}

{
Quantum states of angular momentum and spin generally are not invariant under rotations of the reference frame. Therefore, they can be used as a resource of relative orientation, which is encoded in the asymmetry of the state under consideration. In this paper we introduce the analytical characterization of the rotational information by parameterizing the group characteristic function by polynomial functions. By doing so, we show that the set of states achievable through transformations lacking a reference frame (rotationally covariant ones) admits an analytical characterization and can be studied through the use of semidefinite optimization techniques. We demonstrate the developed methods via examples, and provide a physical scenario in which a reference-independent operation performs a metrologically useful operation: the preparation of a state of light improving interferometer sensitivity, which equivalently can be realized as a postprocessing step.
}
\end{abstract}

\maketitle

\section{\label{sec:Intro}Introduction}

\begin{figure}[t]
    \centering
    \includegraphics[width=\linewidth]{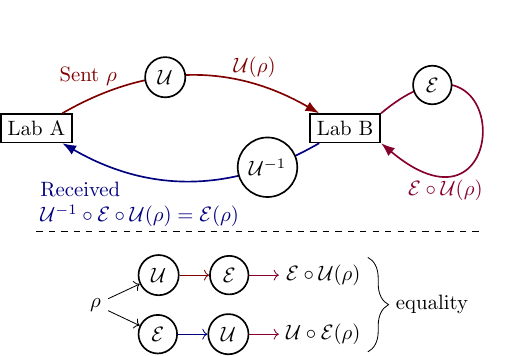}
    \caption{{The constraint of a quantum channel being symmetric with respect to a certain symmetry group is called} covariance and can be interpreted in the two pictured ways. \emph{Top}: a state $\rho$ is transferred between two laboratories, but as a result of lack of common reference the laboratory B receives a state modified by a unitary $U$ describing a particular element of a symmetry group. The channel $\EE$ is applied to the modified state $\mathcal{U}(\rho)=U\rho U^\dagger$ and it is sent back; the result from perspective of A is equal to $\EE(\rho)$ if the channel $\EE$ is covariant. \emph{Bottom}: equivalently, the channel $\EE$ is covariant if the result of state transformation by sequential application of $\EE$ and $\mathcal{U}$ does not depend on the order of operations (for all $\mathcal{U}$ and $\rho$).
    \label{fig:transfdiagram}}
\end{figure}
Symmetry plays a leading role across different fields of physics. From the most 
basic interactions to effective models, physical systems often exhibit some sort of it: the (a)symmetry of a state the system is found in determines its physical properties; the evolution may be invariant under a particular group of transformations; and yet symmetric states are sometimes unstable, as evidenced by spontaneous symmetry breaking phenomena.

Symmetric evolution implies a type of conserved quantity exists. This is the statement of Noether's theorem, valid for nondissipative systems, but similar observations can be made in the general case. The resource theory of asymmetry studies this type of questions {and \emph{the constraints resulting from symmetric evolution} in the context of state transformations are the main topic of this paper (see Fig. \ref{fig:transfdiagram}).}

The group of rotations is described by $\OTH$, the orthogonal group of order 3, or, in the context of quantum spin states, the closely related $\SUT$, the special unitary group of order 2. Here, the resource under consideration is the directional reference: the amount of orientation information a system provides. Systems found in rotationally invariant states convey no such information: they are singlet states, maximally mixed states of definite angular momenta, and their probabilistic combinations. Physically, no quantization axis is distinguished by such states, so they can not be used to convey directional reference. Any other state can -- in multiple inequivalent ways, and some states are better suited for that than others. For instance, two spin-$\frac12$ particles can transmit directional information either through a symmetric pair of parallel spins along the direction to be transmitted $\ket{\uparrow\uparrow}$ or as an antiparallel one $\ket{\uparrow\downarrow}$ -- the latter choice provides a better quality reference than the former spin-coherent state \cite{Gisin1999}.

Systems can also be symmetric with respect to time evolution. Harmonic evolution exhibits such a behavior: it is periodic, hence symmetric with respect to discrete time translations, and the relevant group of time translations becomes isomorphic to $\UO$, the unitary group of order 1. The information held in the state can be viewed as a time or relative phase reference across distant systems.

Questions arising from this type of structure have been studied before. In a two party scenario, a reference frame can be established via state transfer, and optimal protocols are known \cite{Gisin1999,Chiribella2005,BAGAN2006}. {The amount of information that can be gained by the measurement of a state modified by an unknown transformation is directly tied to the properties of the input state, and can be quantified \cite{Yang2017}.} This is the only way to do so: the reference needs to take a physical form, and the information is \emph{unspeakable} \cite{peres2002unspeakable}. {Consequences of relations between different reference frames include such foundational aspects of quantum theory as interpretation of Wigner's friend {thought} experiment \cite{de_la_Hamette2020}.}

Finally, establishing a reference frame between parties might be infeasible (e.g., it fluctuates rapidly and maintenance would have to be performed too often). It is known that even with this limitation, a state prepared with respect to unknown frame might be manipulated to some extent \cite{Gour_2008}, but the allowed operations are known to only decrease the informational content of the state \cite{Iman2012,Marvian2014Modes}. {The lack of a reference frame effectively leads to additional noise, reducing the efficiency of quantum metrology schemes \cite{afr_nek2015}.}

The meaning of frame-independent operations can be interpreted in the following way. Consider two labs, A and B (\cref{fig:transfdiagram}). If the relation between reference frames between the labs was known, the mathematical description  of a state (amplitudes of a pure state, matrix elements of a mixed state) with reference to A could be translated to B via an operation $\mathcal{U}$ (system-dependent, e.g. phase shift, rotation of coordinate system). However, if the relation is not determined (phase shift fluctuates, optical fiber modes mix), the operation $\mathcal{U}$ is not known. Therefore, if B was to transform the received state via a channel $\EE$ and send it back through the same quantum link, the only deterministically realizable operations from A's perspective are the ones that act the same regardless of $\mathcal{U}$. Such operations are called \emph{covariant} with respect to the group of frame transformations in question.

{Equivalently, the \emph{covariant} transformations can be thought of as commuting with all relevant $\mathcal{U}$: in this view, the preparation step with the channel $\EE$ with subsequent unknown transformation $\mathcal{U}$ is equivalent to the application of $\mathcal{U}$ first and postprocessing with $\EE$. One of the consequences of this property is explored in \cref{sec:interferometer}, where we show that in a balanced interferometer, the metrological properties of a coherent light state can be (probabilistically) improved with such an operation. In this scenario, if $\EE$ is applied first, it can be interpreted as the \emph{preparation of squeezed vacuum}, but it is possible to use it in the postprocessing step with the same overall result. }

Abstract classical and quantum information can be transferred {even without reference}: e.g., the amplitudes defining a qubit state may be encoded in a two-dimensional subspace invariant to the relevant transformation group \cite{Bartlett2007}. {A quantum state can also be sent along with a finite-size reference state, and a joint measurement of both parts leads to a better communication efficiency than the sequential procedure of establishing the reference frame and quantum communication \cite{Bartlett2009}. Similarly, if the information is encoded as a relation between a \emph{pair of systems} prepared with respect to an unknown reference frame, the optimal estimation procedure involves entangled measurements of both parties \cite{Bartlett2004}.}

The allowed transformations can be characterized for arbitrary groups \cite{Iman2012,Marvian2014Charac}, but the form of the transformation criterion may be unwieldy for nonabelian groups \cite{szymanski2023numerical}. This is the case for the group of rotations, and due to this mathematical complexity only limited results are known \cite{Gour_2008,szymanski2023numerical}. 

{This article addresses the question of state transformation with rotationally covariant channels, and with this in mind we developed}  the analytical characterization of the $\SUT$ characteristic functions {(\cref{sec:polychar})} -- a concept stemming from the polynomial description of $\UO$ \cite{hall2003lie,szymanski2023numerical}. It allows for a direct answer to the question of pure state interconversion by utilization of the algebra of complex homogeneous polynomials. The main results are contained in \cref{sec:covtrafos}. {In \cref{sec:interferometer} we present a possible application of our formalism, by proving that metrological properties of light states can be probabilistically improved with operations commuting with the action of an interferometer.} 

In order to provide mathematical and physical context of our research, in \cref{sec:theo_background} we review known results with the exemplary use of $\UO$ -- the group representing the time evolution of a harmonic system -- as a toy example for the theory.
\section{Theoretical background}\label{sec:theo_background}
\subsection{Phase reference and $\UO$ group}
Consider a standard quantum harmonic oscillator system, described by the Hamiltonian
\begin{equation-aligned}
    H = a^\dagger a,
    \label{eq:harmham}
\end{equation-aligned}
and denote the eigenstate of the energy  $n\in \NN$ by $\ket{n}$. An arbitrary {pure} state can be written as a sum 
\begin{equation}
\ket\psi=\sum_{n\in\NN} \psi_n \ket{n},
\label{eq:stateuone}
\end{equation}
 and its evolution dictated by \cref{eq:harmham}  is described by the unitary operator {$U(t)$ via}
\begin{equation-aligned}
    \ket{\psi(t)} = \!\!\overbrace{U(t)}^{\negphantom{xxxxxx}\sum_{n\in \NN} \exp(-i n t) \ket{n}\bra{n}\negphantom{xxxxxx}}\!\! \ket\psi=\sum_{n\in\NN} \psi_n \exp(-int) \ket{n}.
     \label{eq:unituone}
\end{equation-aligned}
The system described by \cref{eq:harmham} is harmonic and its evolution is cyclic: $U(t)=U(t+2\pi)$. Evidently, $U(t)$ is described by just a single phase $\exp(-it)$. The set of all such phase shifts $\{\exp(-it) \in \mathbb{C} : t\in[0,2\pi)\}$ is the definition of the group $\UO$, and $U(t)$ can be interpreted as a representation of this group. This means that harmonic evolution realizes abstract phase shifts in a physical system.

The eigenstates $\{\ket{n}\}$ evolve only trivially: they gain a phase factor, which without any other input is undetectable. Any measurement can only detect a \emph{phase difference}, either to external reference or -- as in the case of Aharonov-Bohm effect -- another part of the same system undergoing a different evolution. The eigenstate could be therefore considered as \emph{maximally symmetric states} with respect to the $\UO$ group. The study of state (a)symmetry with respect to different groups has been considered before \cite{Gour_2008, Iman2012,Marvian2014,Marvian2014Charac,Marvian2014Modes}; in this section we present the main already known results in order to provide a theoretical background for our findings.

Suppose the state $\ket\psi$ of \cref{eq:stateuone} was prepared with respect to an unknown phase reference. Such a situation could arise if the state describes a quantum state of a single mode of light, which was prepared in a distant lab (see \cref{fig:transfdiagram}): the lab A created $\ket\psi$ and sent it to lab B, but from its perspective any $\ket{\psi(t)}$ (\cref{eq:unituone}) is equally likely. For some purposes this means that the received state can be described by a density operator \cite{BARTLETT2006}
\begin{equation}
    \rho=\sum_{n\in\NN} \lvert \psi_n\rvert^2 \ket{n}\bra{n},    
\end{equation}
but this is not the entire picture: if the same state is simply reflected back to the lab A unaltered, it is still coherent. 

Can lab B perform any \emph{nontrivial} operation maintaining it coherence? The answer is affirmative, if and only if the operation is reference-independent. If lab $B$ applies a unitary ${T}$ to the received state and subsequently sends it back, from the perspective of lab A it is described by $U(t)^{-1} {T} U(t)$ -- frame change, the applied operation, and final inverse frame change back to the original one. Since $U(t)$ is unknown, for the realized operation to be frame-independent, $U(t)^{-1} {T} U(t)={T}$, which is equivalent to $[{T},U(t)]=0$ for all $t$.

Most general quantum operations are mathematically described by a unitary interaction ${T}$ involving an environment (or auxiliary system), which is then discarded since -- by assumption -- it later never interacts with the system of interest. Such operations may create statistical mixtures of pure states, which are typically described by density operators: if an operation produces a state $\ket{\psi_i}$ with probability $p_i$, the corresponding density operator is $\sum_i p_i \ket{\psi_i}\bra{\psi_i}$.

In this view, quantum operations are linear maps $\EE$: here, such an operator would take an arbitrary density operator $\rho=\sum_{m,n\in \NN} \rho_{m,n} \ket{m}\bra{n}$ as an input and return a modified state $\rho'$ of similar form.

Which of such most general operations could be deterministically performed on the state even without the phase reference? Similarly to the restricted case of unitary transformations (where $[{T},U(t)]=0$ for all t), a quantum channel should commute with phase shifts, in the following sense.

Let us denote the change of phase reference by an unitary channel $\UU_t(\rho) = U(t) \rho U(t)^\dagger$. Whatever quantum channel $\EE$ is applied, the result of the operation should not depend on the (unknown) reference: $\UU_t^{-1} \circ \EE \circ \UU_t$ should not depend on $t$. This results in the \emph{covariance} criterion for realizable quantum channels $\EE$, presented below. For the later use, it is formulated for an arbitrary group $G$ -- here, the relevant group is $\UO$.
 
\begin{defi}
A quantum channel $\EE$ is group-covariant with respect to the group $G$ with representation $\UU_g$ if and only if for all $g\in G$,
\begin{equation}
 \EE \circ \UU_g = \UU_g \circ \EE.
\end{equation}
\end{defi}

In short, the operations have to commute with the action of the symmetry group. This is a fundamental limitation independent from experimental feasibility: a complex, but still covariant, quantum channel might require an improvement in experimental techniques for its implementation. On the other hand, even very simple non-covariant operations can not be realized at all.

A covariance requirement independently arises in other contexts: it is a part of the constraints for thermal operations \cite{Lostaglio2015,Lostaglio2015A} and is a result of fundamental symmetries of nature, through axiomatic superselection rules \cite{Gour_2008}.

For the phase shift group $\UO$, the characterization of covariant operations is known \cite{Gour_2008}. Therefore, it is possible to classify states which can be prepared from a given  $\ket\psi$ if its phase reference is unknown. For instance, the phase of each eigenstate can be manipulated independently as a $\UO$-covariant operation; therefore, any state defined by \cref{eq:stateuone} can be turned to the canonical form of
\begin{equation}
    \ket\psi = \sum_{n\in\NN} \sqrt{p_n} \ket{n},
    \label{eq:defpsiu1}
\end{equation}
with $p_n = \lvert \psi_n \rvert^2$. Such a shift would require the engineering of a nonlinear addition to the Hamiltonian \cref{eq:harmham}, but is independent of a phase reference.

Because of this reduction, let us concentrate on states with real, positive amplitudes. It is known \cite{Gour_2008} that the state $\ket\psi$ defined above can be turned into another if and only if $p_n$ can be split into a convolution of two probability distributions over $\NN$ {after a constant shift (see \cref{fig:u1diag} for an example with $\delta=0$):}
\begin{prop}
    The $\UO$-covariant quantum channel $\EE$ such that $\sum \sqrt{p_n} \ket{n}
    \xmapsto{\EE}    \sum_{n\in\NN_0} \sqrt{q_n} \ket{n}$ exists if and only if there exists $\delta\in\NN$ and a sequence $w_n$ such that $\sum_{n\in\NN} w_n=1$, $w_n\ge0$ and
    \begin{equation-aligned}
    {p_{n-\delta} = \sum_{n\ge m\ge 0} q_m w_{n-m}.}
\end{equation-aligned}
\label{prop:uocov}
\end{prop}
\begin{figure}[t]
    \centering
    \includegraphics[width=\linewidth]{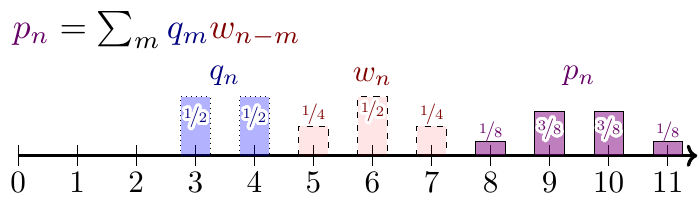}
    \caption{Example of probability distributions compatible with \cref{prop:uocov}: deterministic transformation of $\ket\psi=\sum_n \sqrt{p_n} \ket{n}$ to $\ket\phi=\sum_n \sqrt{q_n} \ket{n}$ is possible  with a $\UO$-covariant channel, because $p_n$ is a convolution of $q_n$ with an auxiliary $w_n$. 
    \label{fig:u1diag}}
\end{figure}
This condition can be cast in the language of polynomial theory \cite{szymanski2023numerical}: it is exactly the equation describing the coefficients $p_n$ of a product of univariate polynomials
\begin{equation-aligned}
    {\sum_{n\in\NN} p_n z^{n+\delta} = \left(\sum_k q_k  z^k\right)\left(\sum_l  w_l z^l \right).}
    \label{eq:u1polydecomp}
\end{equation-aligned}
Via this observation, it is straightforward to verify several properties of the accessible states. For arbitrary $\ket{\psi}$, only finitely many pure states are accessible, up to a shift in energy. The variance of the energy of the resulting state cannot increase \cite{Marvian2014}, and for generic probability distributions $p_n$ only trivial operations (energy shifts) are admissible, since no decomposition of the form \cref{eq:u1polydecomp} exists with valid (nonnegative) coefficients {$w_n$}.

\subsection{General results}

The general theory of group-covariant transformations for \emph{arbitrary} groups was developed in \cite{Iman2012} and articles following the thesis \cite{Marvian2014,Marvian2014Charac,Marvian2014Modes}. It is an abstract characterization valid for any group $G$; the central point of the theory is the so-called \emph{characteristic function}. 
\begin{defi}
Consider a group $G$ with representation $U_g$. The \emph{characteristic function} of a state $\ket\psi$ belonging to the space of the representation is the function $\chi: G \rightarrow \CC$,
\begin{equation}
    \chi_\psi(g)=\bra{\psi}\!U_g\!\ket{\psi}\!.
\end{equation}
\label{def:characfun}
\end{defi}

The characteristic function captures the entire asymmetry information of the state with respect to the group, and in certain cases (e.g., $\ket\psi$ is contained within a single irreducible representation) contains the entire information needed to reconstruct the state. What has been found in~\cite{Iman2012,Marvian2014Charac} is that the characteristic function $\chi_\psi$ determines the set of covariantly accessible pure states $\ket\phi$. {We present a version of the propositions in~\cite{Iman2012}, modified for clarity and tailored to the requirements for $\SUT$-covariant transformations. For simplicity, we assume that the group $G$ has no nontrivial one-dimensional irrep; in order to avoid mathematical inconsistencies, one might also require that the characteristic functions $\chi$ are sufficiently regular. The simplest (but most restrictive) regularity condition is that all appearing states belong to a given finite-dimensional Hilbert space $\HH$ (which may contain arbitrarily many trivial one-dimensional irreps of $G$); commonly encountered infinite-dimensional states (e.g. coherent states in the case of $\UO$) cause no problems as well.}\\
The equality of characteristic functions implies the existence of two-way covariant transformations:

\begin{prop}\label{prop:directreach}
    Consider two pure states $\ket\psi$ and $\ket\phi$. There exists a $G$-covariant unitary ${V}$ such that ${V}\ket\psi=\ket\phi$ and ${V}^\dagger\ket\phi=\ket\psi$ if and only if the characteristic functions of the states are equal:
\begin{equation}
    \chi_\psi(g)=\chi_\phi(g),~~\forall g\in G.
\end{equation}
\end{prop}

This serves as a basis for a one-way conversion criterion: via the Stinespring dilation of $G$-covariant channels,  a $G$-invariant auxiliary state $\ket\eta$ ($\chi_\eta(g)=1$ for all $g\in G$) can be added, followed by the application of a $G$-covariant unitary and possibly partial trace. The system traced over may carry some information; together with the target state, it should have the same symmetry properties, leading to the following characterization of pure state transformations (one of the main results of \cite{Marvian2014Charac}):
\begin{prop}\label{prop:iman1}
	There  exists a $ G $-covariant channel $ \EE $ mapping a pure state $ \ket{\psi} $ to $ \ket{\phi} $ if and only if there exists a state $ \ket\xi $ such that for all $ g\in G $: 
\begin{equation} 
\chi_\psi(g)=\chi_\xi(g)\chi_\phi(g).
\end{equation}
\end{prop}

Further analysis shows that the \emph{probabilistic generation} is described by characteristic functions as well. They are realized in a similar fashion as the deterministic transformations: a $G$-invariant auxiliary state is added, followed by the application of a $G$-invariant unitary and a measurement of the auxiliary state in a $G$-invariant basis. Depending on the outcome of this measurement, a part of the system is traced out to end up with an ensemble $\{\ket{\phi_i},p_i\}_i$. This can be summarized as follows (Theorem~65 in \cite{Iman2012}).

\begin{prop}\label{prop:stochastic_trafo}
	There exists a $ G $-covariant map transforming $ \ket\psi $ to the ensemble $ \{\ket{\phi_i},p_i\}_i $ if and only if there are states $ \{\ket{\xi_i}\}_i $  such that for all $ g\in G $: 
 \begin{equation}
     \chi_\psi (g)=\sum_i p_i\chi_{\xi_i}(g)\chi_{\phi_i}(g).
 \end{equation}
\end{prop}
Note that the output of the channel is a postselected pure state and not the mixed state $\rho=\sum_i p_i \kb{\phi_i}{\phi_i}$. If only the target state is of interest, the result can be modified as shown in Corollary 66 in \cite{Iman2012}:
\begin{cor}\label{corr:stochastic_trafo}
	There exists a $ G $-covariant map transforming $ \ket\psi $ to $ \ket\phi $ with probability $ p $ if and only if there are states $ \ket\xi$ and $\ket\sigma $ such that for all $ g\in G $: 
 \begin{equation}
     \chi_\psi (g)=p\chi_\xi(g)\chi_\phi(g)+(1-p)\chi_\sigma(g).
 \end{equation}
\end{cor}
{These findings, after minor modifications }\footnote{There exist nontrivial one-dimensional representations of $\UO$: every Fock state $\ket{n}$ corresponds to one.} {that allow for the structure of $\UO$}, are consistent with the earlier results concerning this group. In particular,  the characteristic function of the state \cref{eq:defpsiu1} is the Fourier transform of the defining probability distribution:
\begin{equation-aligned}
    \chi_\psi(t) = \sum_{n\in\NN} p_n \exp(-i n t),
    \label{eq:uocharfun}
\end{equation-aligned}
which can be interpreted as one of the polynomials appearing in \cref{eq:u1polydecomp} with $z:=\exp(-it)$.
The product of characteristic functions appearing in \cref{prop:iman1}, after an inverse Fourier transform, corresponds to the convolution of two probability distributions -- a result equivalent to \cref{prop:uocov}.

\subsection{Three-dimensional rotations: $\SUT$}

Physical reality does not depend on the chosen coordinate system, but \emph{some} must be chosen in order to accumulate experimental results. Frequently, a Cartesian coordinate system is used to refer to spatial degrees of freedom: any point is described by $\vec r:=(x,y,z)\in\RR^3$. The three numbers implicitly assume some directional reference; another observer may prefer a different one, with the transformation defined as $\vec r\mapsto O \vec r+\vec r_0$, where $O$ is an orthogonal matrix ({that is,} $O O^T$ is an identity matrix) and $r_0$ is a constant coordinate shift. 

Here we are interested in the lack of a \emph{rotational} frame reference and its effect on quantum operations. For this, a description of how a rotation of the reference frame affects the mathematical description is needed. Naturally, a wavefunction of a single spinless massive particle (e.g. electron in the original Schr\"odinger work \cite{Schr_dinger1926}) transforms as $\psi(\vec r)~\mapsto~\psi'(\vec r):=\psi(O \vec r)$. If a state is expanded in the basis of definite angular momenta $j$, and its projection on the $z$-axis (plus auxiliary indices $\alpha$ for any other degrees of freedom preserved by rotations, e.g. the principal quantum number):
\begin{equation-aligned}
    \ket\psi = \sum_{j=0,1,\ldots} \sum_\alpha \sum_{m=-j}^j \psi_{j,m,\alpha} \ket{j,m,\alpha},
\end{equation-aligned}
the effect of a coordinate change can be summarized as
\begin{equation-aligned}
    \ket{\psi'} = \sum_{j=0,1,\ldots} \sum_\alpha \sum_{m,m'=-j}^j U^{(j)}_{m,m'}\psi_{j,m',\alpha} \ket{j,m,\alpha}.
\end{equation-aligned}
The matrices $U^{(j)}$ form a \emph{representation} of the group of rotations: they describe the effect of coordinate change on the mathematical description of the state. 

The group of rotations is the main subject of this paper, and the above description is sufficient for orbital angular momenta. To accommodate the spin degrees of freedom (which can not be expressed as coordinate change in wavefunction), in calculations we will use the related group $\SUT$, which is the set of complex matrices
\begin{equation-aligned}
\SUT:=\left\{\begin{pmatrix} u&-v^*\\ v & u^*\end{pmatrix}: \lvert u\rvert^2+\lvert v\rvert^2=1\right\}.
\label{eq:su2def}
\end{equation-aligned}
 Any matrix from $\SUT$ can be transformed to an orthogonal coordinate change $O$, and subsequently, the action on amplitudes through $U^{(j)}$, but it provides a bit more generality and a more straightforward mathematical description. Representations of this group can be identified with spin and angular momentum states{; explicitly, the matrices in \cref{eq:su2def} can be though of as rotations of spin-$\frac12$ with $V=\exp\left(i(n_x \sigma_x +n_y \sigma_y+n_z \sigma_z)\right)\in\SUT$, where $\sigma_i$ are the standard Pauli matrices}. Our aim is to provide exhaustive criteria and numerical methods for the problem of state interconversion, that is the problem \emph{which quantum channels $\EE$ do not depend on spatial orientation, and what states can be reached with them?}. So far only limited results were known; one of them is the characterization of Kraus operators that can be used to build $\SUT$-covariant channels (Lemma 17 in \cite{Gour_2008}):

\begin{prop}
    \label{prop:kraussu2}
    Any $\SUT$-covariant quantum channel $\EE$ allows the decomposition $\EE(\rho) = \sum_i K_i \rho K_i^\dagger$ into Kraus operators with the most general form of $K_i$ as
 \begin{equation-aligned}
 \label{eq:krausformsu2}
     &K_{J,M,\alpha}(\{\vec{f}_{J,\alpha}\})=\sum_{j^\prime=0,\frac{1}{2},1,\dots}\sum_{m=-j^\prime}^{j^\prime}\sum_{j=\abs{J-j^\prime}}^{J+j^\prime}\\&~~\begin{pmatrix}
 		j^\prime & J & j\\ \!\!-m\! & \!\!M & \!\!m\!-\!M
 	\end{pmatrix}\!(-1)^{j-m}\!f_{J,\alpha}^{(j^\prime,j)}\!\!\kb{j^\prime,m}{j,m\!-\!M},
 \end{equation-aligned}
where $ J\in\{0,\frac{1}{2},1,\dots\} $, $ M\in\{-J,-J+1,\dots,J\} $,  $ \alpha $ is a multiplicity index, {the Wigner $3j$-symbol related to the Clebsch-Gordan coefficients is denoted by}
\begin{equation}
    \begin{pmatrix}
 		j^\prime & J & j\\ \!\!-m\! & \!\!M & \!\!m\!-\!M
 	\end{pmatrix}\!,
\end{equation}  and the vector $ \vec{f}_{J,\alpha} $ has entries $ f_{J,\alpha}^{(j,j^\prime)} $.
The standard normalization condition
\begin{equation}\label{eq:normkraus}
	\sum_{J,M,\alpha}K_{J,M,\alpha}(\{f_{J,\alpha}\})^\dagger K_{J,M,\alpha}(\{f_{J,\alpha}\})=\ID,
\end{equation}
can be written as
\begin{equation}\label{norm_channel}
	\sum_{J,j^\prime,\alpha}\abs{f_{J,\alpha}(j^\prime,j)}^2=2j+1.
\end{equation}
\end{prop}

$\SUT$-covariant unitary channels have just one Kraus operator $K_{0,0,0}=\sum_j e^{i\theta_j}\Pi_j$, with $\Pi_j$ being projectors onto the irrep $j$, changing the relative phases between the irreps.\\
This has led to partial answers to the transformation problems \cite{Gour_2008}: if we restrict ourselves to superpositions of spin-coherent states in a single direction $\Vec{n}$, the set of accessible states admits a simple characterization. Without loss of generality, we can assume that $\Vec{n}=(0,0,1)^T$ and the states take the form
\begin{equation}
    \ket\psi=\sum_j \sqrt{p_j} \ket{j=j,m=j},
    \label{eq:sutcohst}
\end{equation}
with $\{p_j\}_{j=0,\frac12,1,\dots}$ being a probability distribution. The following proposition (Theorem 20 in \cite{Gour_2008}) solves the interconversion problem for this restricted subset of spin states.
\begin{prop}
Consider a pair of states defined as in \cref{eq:sutcohst}: the state $\ket\psi$ defined by the probability distribution $p_j$ and $\ket\phi$ defined by $q_j$. There exists a $\SUT$-covariant channel $\EE$ such that $\EE(\ket\psi)=\ket\phi$ if and only if there exists a probability distribution $\xi_j$, such that
\begin{equation}
    p_j=\sum_{J=0,\frac12,1,\dots} \xi_J q_{j+J}.
\end{equation}
\end{prop}
Note the similarity with \cref{prop:uocov}. It is not accidental; for this class of states the most relevant part of the $\SUT$ group is rotation along the quantization axis, again described by $\UO$. Here, a shift of the prior probability distribution $p_j$ is not allowed, because the only pure invariant state, which can be added as an ancilla in the Stinespring dilation of the channel $\EE$, is $\ket{0,0}$ (plus multiplicities). Further results concerning stochastic transformations can be found in \cite{Gour_2008}.

\section{\label{sec:polychar} Relevant mathematical structures}

\subsection{Polynomial $\SUT$ representation}
\label{subsec:majorana}

In this section, we will show how the action of $\SUT$ rotations on spin states can be parameterized using polynomials. The parametrization is extracted from the form of overlap functions $\langle\psi\vert\vec n\rangle$, where $\ket{\vec n}$ is a spin coherent state; this is related to the \emph{Majorana stellar representation} (see \cref{app:majorana}), which allows for an equivalent description of a state $\ket\psi$ by points $\vec n$ on a sphere for which this overlap vanishes. Here, a particularly simple form of a coherent state transformation under rotations (\cref{eq:cohtrafosu2}) is employed to extract the polynomial form of the rotation from the overlap function (\cref{eq:majdefeq}), which then can be applied to an arbitrary state $\ket\psi$ (\cref{prop:su2_reps}).
The reasoning starts with pure states of definite total angular momentum: 
\begin{equation-aligned}
    \ket\psi = \sum_{m=-j}^j \psi_m \ket{j,m}.
    \label{eq:defjstate}
\end{equation-aligned}
{The } \emph{spin coherent states} {used in the derivations maximize the spin component along a certain axis} {defined as follows} \cite{Chryssomalakos2018}:
\begin{defi}
\label{def:coherentspin}
    The spin coherent state $\ket{\vec n}$ within an irreducible representation corresponding to the total angular momentum of $j$ (with the spin matrices $J_x, J_y, J_z$) is the normalized eigenstate of the maximum eigenvalue of $\vec n \cdot \vec J$:
    \begin{equation-aligned}
        (\vec n \cdot \vec J)\ket{\vec n} = j \lvert \vec n \rvert \ket{\vec n}.
    \end{equation-aligned}
\end{defi}
A simple example of a coherent state is $\ket{j=J,m=J}$ for any $J$: it corresponds to the vector $\vec n=(0,0,1)$. Such states, are defined up to a phase only, but this does not pose an issue here, since we employ a concrete parametrization stemming from the observation that $\ket{\vec n}$ can be regarded as a symmetric product state of $2j$ spin-$\frac12$ subsystems:
\begin{equation}
\label{eq:cohasprod}
\ket{\vec n}=\left(z_1\ket{\uparrow}+z_2\ket{\downarrow}\right)^{\otimes 2j},
\end{equation}
where the complex numbers $z_1, z_2$ determine the state of each individual spin-$\frac12$ subsystem. Back to the spin-$j$ representation, it can be written in full as 
\begin{equation}
\ket{\vec n}=\sum_{m=-j}^j \sqrt{\binom{2j}{j-m}}z_1^{j+m}z_2^{j-m}\ket{j,m}.
\label{eq:cohdefpoly}
\end{equation}

Provided that $\lvert z_1 \rvert^2 + \lvert z_2 \rvert^2=1$ (which amounts to normalization of $\ket{\vec n}$), the expectation values of the spin operators can be shown to be
\begin{equation-aligned}
    \langle J_x \rangle_{\vec n} \!=\! 2j& \Re(z_2 z_1^*),~~ \langle J_y \rangle_{\vec n}\!=\! 2j \Im(z_2^\ast z_1),\\
    &\langle J_z \rangle_{\vec n}=j (\lvert z_1\rvert^2 - \lvert z_2 \rvert^2).
    \label{eq:expvals}
\end{equation-aligned}
Now we can define a homogeneous polynomial $f_\psi$ in complex variables $z_1, z_2$ associated with the state $\ket\psi$ via the overlap
\begin{equation-aligned}
    f_\psi(z_1,z_2) &= \langle\psi \vert \vec n \rangle \\&
    =\!\! \sum_{m=-j}^{j}\!\sqrt{\binom{2j}{j-m}}\psi_{m}^* z_1^{j+m}z_2^{j-m}.
    \label{eq:majdefeq}
\end{equation-aligned}
If a state $\ket\psi$ undergoes a rotation with 
\begin{equation-aligned}
U={\exp\left(i (n_x J_x^{(j)}+n_y J_y^{(j)}+n_z J_z^{(j)})\right)},
\label{eq:expsu2param}
\end{equation-aligned}the homogenous polynomial of $U\ket\psi$ transforms as
\begin{equation-aligned}
    f_\psi = \langle \psi \vert\vec n\rangle \mapsto f_{U\psi} = \overbrace{\langle \psi \vert U^\dagger}^{(U\ket\psi)^\dagger}\! \vert \vec n \rangle.
\end{equation-aligned}
Importantly, one can act with $U^\dagger$ on the coherent state itself; evidently, it amounts to a rotation of each individual spin-$\frac12$ component in \cref{eq:cohasprod}. The action on the spin-$\frac12$ subspace defines the rotation completely; in accordance with the definitional requirements of the group $\SUT$ (that they are unitary and unit determinant complex matrices of size 2) one can parameterize the action as 
\begin{equation-aligned}
    U^{(j=\frac12)}=\overbrace{\begin{pmatrix} u & -v^\ast \\ v & u^\ast \end{pmatrix}}^{V},
    \label{eq:su2param}
\end{equation-aligned}
and in consequence, a state $\ket{\vec n}$ defined by complex numbers $z_1, z_2$ is transformed to $U^\dagger \ket{\vec n}$, corresponding to $z'_1, z'_2$ by
\begin{equation-aligned}
     \begin{pmatrix}
     z'_1\\z'_2
     \end{pmatrix}:= \overbrace{\begin{pmatrix}
         u^* & v^*\\-v&u
     \end{pmatrix}}^{V^\dagger} \begin{pmatrix}
     z_1\\z_2
     \end{pmatrix}.
     \label{eq:cohtrafosu2}
\end{equation-aligned}
This simple transformation of the coherent state parameters $z_1, z_2$ under rotations can be translated to a strict statement (see Chapter 4.3.4 of \cite{hall2003lie}).
\begin{prop}\label{prop:su2action}
Let us parameterize the element $V$ of $\SUT$ as in \cref{eq:su2param}. Take a state defined as in \cref{eq:defjstate} corresponding to a polynomial $f_\psi$ defined in \cref{eq:majdefeq}. Then, the state $\ket\phi=U\ket\psi$ transformed by a standard representation corresponding to spin $j$ has the polynomial
\begin{equation}
    \begin{aligned}
        &f_{\phi}(z_1,z_2)= f_\psi\left(V^{\dagger} \begin{pmatrix}z_1 \\ z_2 \end{pmatrix}\right) \\
        &=f_\psi(\overbrace{u^\ast z_1+v^\ast z_2}^{z_1'},\overbrace{-vz_1+u z_2}^{z_2'})\\
        &=\sum_{m=-j}^{j} \sqrt{\binom{2j}{j-m}}\psi_{m}^* z_1'^{j+m}z_2'^{j-m}.
    \end{aligned}
    \label{eq:polyaction}
\end{equation}
\end{prop}

As a result, the action of $\SUT$ on a space of spin-$j$ states can be parameterized by polynomials in $u, u^*, v,$ and $v^*$: this is a basis of the results in the following section. The polynomial representation $U^{(j)}$ can be read off from the transformation laws for monomials appearing in \cref{eq:polyaction}, {and} the matrix element {$U^{(j)}_{m^\prime,m}=\langle j,m'\vert U^{(j)} \vert j,m\rangle$} is given by 
\begin{equation-aligned}
    \!\!\!U^{(j)}_{m',m}\!=&\!\sqrt{\!\frac{\!\binom{2j}{j-m}}{\!\binom{2j}{j-m'}}}\sum_{\!\!a=\max \left\{0,m'-m\right\}}^{\!\!\min\left\{\!j\!-\!m\!,j\!+\!m'\!\right\}}\!\!\!\!\binom{\!j\!-\!m\!}{a}\!\!\binom{j\!+\!m}{\!m\!\!-\!m'\!+\!a\!}\\
    &\times (-1)^{a} u^{j+m'-a}u^{\ast j-m-a}v^{m-m'+a} v^{\ast a}.
    \label{eq:rep_polysum}
\end{equation-aligned}
 The inside sum can be calculated to yield an ordinary hypergeometric function $\,_2 F_1$ -- see \cref{app:majorana} for details; in the following text the hypergeometric form of $U^{(j)}$ is not going to be used.

\begin{prop}\label{prop:su2_reps}
    The coefficients $\{\psi_m\}_{m=-j}^j$ of a spin-$j$ state transform under a rotation by $V=\begin{pmatrix} u & -v^\ast \\ v & u^\ast \end{pmatrix}\in\SUT$ as {$\psi_{m'}=\sum_{m=-j}^j U^{(j)}_{m',m}\psi_m$ with $U^{(j)}_{m',m}$} given by \cref{eq:rep_polysum}. 
\end{prop}
For instance, the low spin polynomial parametrizations (in the basis of $\{\ket{j,m=j}, \ldots, \ket{j,m=-j}\}$) are:

\begin{equation-aligned}\label{eq:polyrepsu2}
    U^{(j=\frac12)}=&\begin{pmatrix}
        u & -v^\ast \\ v & u^\ast
    \end{pmatrix}~,\\
    U^{(j=1)}=&\left(
\begin{array}{ccc}
 u^2 & -\sqrt{2} u v^\ast & v^{\ast 2} \\
 \sqrt{2} u v & u u^\ast- v v^\ast  & - \sqrt2 u^\ast v^\ast
   \\
 v^2 & \sqrt{2} u^\ast v & u^{\ast 2} \\
\end{array}
\right)~.
\end{equation-aligned}
This form is easier to algebraically manipulate than the standard exponential parametrization (\cref{eq:expsu2param}); it contains the same information, but in explicit form. Since for characteristic functions (\cref{def:characfun}) one has to calculate expressions of form $\langle \psi \vert U\vert\psi\rangle$, and compare them to {ascertain} state reachability (\cref{prop:iman1} and \ref{prop:stochastic_trafo}), this is a significant simplification: \emph{equality of polynomials can be determined by comparing their coefficients} (after taking into account the constraint $\abs{u}^2+\abs{v}^2=1$), while expressions stemming from matrix exponents contain unwieldy trigonometric expressions (see \cref{app:exppar} for an example).

\subsection{Semidefinite programs}
Semidefinite programs are a method for solving a class of convex optimization problems; with their usage it is possible to efficiently numerically optimize a linear function (e.g. maximize the expectation value or probability) over a subset of positive semidefinite matrices (e.g. quantum states with additional constraints, or quantum channels by their dual representation). We will use such methods in later parts of the text; in particular, the result of \cref{prop:fid} with some modifications allows for maximizing the fidelity over a subset of $\SUT$-covariantly reachable states.\\
{For a general definition of semidefinite optimization problem, let us consider two complex vector spaces $\mathcal{X}=\CC^n$ and $\mathcal{Y}=\CC^m$}.  {Now, let $\xi$ transform hermitian operators acting on $\mathcal{X}$ to hermitian operators on $\mathcal{Y}$ and  fix two hermitian operators $C=C^\dagger$  and $D=D^\dagger$ (acting on $\mathcal{X}$ and $\mathcal{Y}$, respectively). With this,} we define the primal problem

\begin{equation-aligned}\label{eq:primal_SDP}
    \max_X \quad &\tr(CX), \\
	\text{s.t.} \quad & \xi(X)=D, \\
	& X \succcurlyeq 0.
\end{equation-aligned}
{Here, $X\succcurlyeq0$
signifies that $X$ is positive semidefinite.
The quality of the optimization of the primal problem can be ensured by the use of a dual problem. It provides an upper bound on the optimal solution. For many classes of problems, the upper bound can be proven to be exact \cite{michalek2021invitation}.

There are many applications of semidefinite programs in quantum information, e.g., it is possible to determine the fidelity between two states, as independently discovered in \cite{Wat2012,Killoran2012}:
\begin{prop}\label{prop:fid}
	The fidelity between two states $ \rho,\sigma\in\operatorname{Pos}(\HH),~\HH=\CC^n $, given by  \begin{equation}\label{key}
		\mathcal{F}(\rho,\sigma)=\tr(\sqrt{\sqrt{\sigma}\rho\sqrt{\sigma}}),
	\end{equation} can be computed with a semidefinite program \cite{Wat2012} with the following primal problem:
	\begin{equation}\label{key}
		\begin{aligned}
			\max_X \quad & \frac{1}{2}\tr(X+X^\dagger), \\
			\text{s.t.}\quad & 
			\begin{pmatrix}
				\rho & X \\ X^\dagger & \sigma
			\end{pmatrix}\succcurlyeq 0, \\
			& X\in \CC^{n\times n}.
		\end{aligned}
	\end{equation}
\end{prop}
{Here, the standard form of the semidefinite program (\cref{eq:primal_SDP}) has been replaced with a simplified one, better suited for this particular case.} {This semidefinite optimization will be used later in \cref{subsec:tramix} for determination of maximum fidelity achievable with rotationally covariant operations.}

\section{\label{sec:covtrafos} {Application to} $\operatorname{SU}(2)$-covariant transformations}

\subsection{Pure state transformations}
The $\SUT$ polynomial representation introduced in \cref{subsec:majorana} can be applied to define the characteristic function in purely analytical terms. Recall that for  a state $\ket\psi=\sum_{m=-j}^j \psi_m \ket{j,m}$ contained in a representation corresponding to the total angular momentum $j$ there is an associated polynomial $f_\psi (z_1, z_2) = \sum_{m=-j}^j f_m z_1^{j+m}z_2^{j-m}$. The group $\SUT$ acts on the polynomials by a linear transformation of the vector of variables $(z_1, z_2)$ -- see \cref{eq:polyaction}. The transformed polynomial $f'$ corresponds to a state $U\ket\psi$ -- thus, the form of the characteristic function $\bra\psi U \ket\psi$ can be read out from the transformed polynomial.
\begin{prop}
The $\SUT$ characteristic function of a state with definite spin $j$,
\begin{equation-aligned}
    \ket\psi = \sum_{m=-j}^j \psi_m \ket{j,m},
\end{equation-aligned}
can be parameterized via two complex variables $u, v$ obeying $\lvert u \rvert^2 + \lvert v \rvert^2=1$ via
\begin{equation-aligned}
    \chi_\psi(u,u^*,v,v^*) = \sum_{m',m\in\{-j,\dots,j\}} U^{(j)}_{m',m}\psi_{m'}^\ast\psi_m,
\end{equation-aligned}
{where the representation $U^{(j)}$ is extracted from the polynomial representation (see \cref{prop:su2_reps} and \cref{eq:rep_polysum}).}
\label{prop:charaspoly}
\end{prop}

If a state is a superposition of states belonging to different irreducible representations, i.e., $\ket\psi = \sum_{j}  \ket{\psi_j}$, the characteristic function $\chi_\psi$ also decomposes:
\begin{equation-aligned}
    \chi_\psi = \sum_{j} \chi_{\psi_j}.
\end{equation-aligned}

If the maximal total angular momentum appearing in $\ket\psi$ is denoted by $J$, the characteristic function of a state $\ket\psi$ is in this context a polynomial of degree $2J$ in $ u,u^\ast,v,v^\ast $ -- the variables characterizing a $\SUT$ group element.
\begin{ex}
	The characteristic function of the state $ \ket\psi=\ket{j=1,m=1} $ is given by 
	\begin{equation}\label{key}
		\chi_\psi(u,u^\ast,v,v^\ast)=u^2.
	\end{equation}
	Similarly, a superposition of multiple irreps  $ \ket\phi=\frac{1}{\sqrt{2}}\left(\ket{j=1,m=1}+\ket{j=0,m=0}\right) $ corresponds to a sum of characteristic functions,
	\begin{equation}\label{key}
		\chi_\phi=\frac12(\chi_{\ket{j=1,m=1}}+\chi_{\ket{j=0,m=0}})=\frac{1}{2}(u^2+1).
	\end{equation}
\end{ex}

{For a strict mathematical formulation of $\SUT$ characteristic functions through their polynomial parametrizations, some care has to be taken. In particular, the exact form of the polynomial might depend on the method of calculation, and a direct comparison of the coefficients is insufficient to determine the equality of characteristic functions (which is needed for \cref{prop:directreach} and subsequent results). Evidently, $\chi_\psi$ and $\chi_\psi+f \times(\abs{u}^2+\abs{v}^2-1)$ are the same functions (with the constraint $\abs{u}^2+\abs{v}^2=1$), but in general are different polynomials. Another problem is practical; when searching for auxiliary states $\ket\xi, \ket\sigma$ {(appearing in \cref{prop:iman1} and \cref{corr:stochastic_trafo})}, the dimension of the Hilbert space must be bounded for computational tractability. A detailed analysis of these problems is presented in the \cref{app:proof_max_j}; here we present the results most important for the understanding of the method.}

The first problem is solved by comparing coefficients of a definite representative of a given characteristic function:
\begin{defi}
    For a given polynomial $\chi$ in variables $u, u^*, v, v^*$, let $\pi[\chi]$ denote a coefficient list 
    \begin{equation-aligned}
        \pi[\chi] = (\tilde\chi_{a,b,c,d})_{a,b,c,d\in \NN},
    \end{equation-aligned}
    such that the value of 
    \begin{equation}
        \tilde \chi=\sum_{a,b,c,d\in\NN} \tilde\chi_{a,b,c,d} u^a u^{*b} v^c v^{*d}
    \end{equation}
    is equal to $\chi$ whenever $\abs{u}^2+\abs{v}^2=1$ and $\tilde\chi$ is not divisible by $\abs{u}^2+\abs{v}^2-1$. This operation is linear in $\chi$, and can be thought of as a projection of coefficients onto a fixed subspace.
    \label{def:coefflist}
\end{defi}

With this definition, we can compare polynomials through the coefficient list provided by $\pi$. Then, the direct application of \cref{prop:iman1} yields the following new result:
\begin{prop}\label{prop:poly_det}
	The deterministic  $\SUT$-covariant  transformation $ \ket\psi\mapsto\ket\phi $ is possible if and only if there is a state $ \ket\xi $ such that
	\begin{equation}\label{key}
		\pi[{\chi}_\psi]=\pi[{\chi}_{\xi\otimes\phi}]~.
	\end{equation}
\end{prop}
\begin{proof}
    See \cref{app:proof_det}.
\end{proof}
Similarly, the application of the polynomial description together with \cref{prop:stochastic_trafo} and \cref{corr:stochastic_trafo} can be summarized by the following Proposition.
\begin{prop}\label{prop:poly_stoch}
	There exists a $\SUT$-covariant map transforming $ \ket\psi $ to $ \ket\phi $ with probability $ p $ if and only if there are states $ \ket\xi$ and $\ket\sigma $ such that
	\begin{equation}\label{key}
		\pi[{\chi}_\psi]=p\pi[{\chi}_{\xi\otimes\phi}]+(1-p)\pi[{\chi}_\sigma]~.
	\end{equation}
\end{prop}
\begin{proof}
    Analogously as for the proof of \cref{prop:poly_det}, the characteristic functions appearing in \cref{prop:stochastic_trafo} can be expressed as polynomials and their canonical versions are compared to yield the result.
\end{proof}

The structures presented in \cref{prop:poly_det} and \ref{prop:poly_stoch} can be used directly to prove the existence of a $\SUT$-covariant transformation if the auxiliary states can be guessed. However, they also allow for numerical methods: the constraints (on polynomial coefficients) they impose are linear, and the search spaces are quantum states. This is exactly the formulation allowing for a semidefinite optimization approach:
\begin{prop}\label{prop:maxprobsdp}
    The maximum probability for a $\SUT$-covariant transformation $ \ket\psi\mapsto\ket\phi $ can be determined with the following semidefinite program:
    \begin{equation-aligned}\label{key}
	    \max_{\rho,\sigma} \quad & \tr(\rho),\\
	   	\text{s.t.}\quad &\pi[{\chi}_\psi]=\pi[{\chi}_{\rho\otimes\phi}]+\pi[{\chi}_\sigma],\\
		& \rho\succcurlyeq0,\\
		& \sigma\succcurlyeq0.
	\end{equation-aligned}
The Hilbert space dimension of $\rho$ and $\sigma$ can be constrained by the maximal total spin appearing in $\ket\psi$ and $\ket\phi$ (see \cref{app:proof_max_j}).
\end{prop}
\begin{proof}
This directly follows from \cref{prop:poly_stoch} with the probability of the transformation as being absorbed in the states $\rho$ and $\sigma$, which are no longer normalized. Their support is restricted with respect to \cref{prop:max_j} in order to ensure a finite-dimensional set to optimize over. 
\end{proof}
This observation can be implemented as a numerical procedure; in the following example, such an approach leads to a result consistent with an analytical prediction.

\begin{ex}\label{ex:pure_state_prob}
	The transformation $ \ket\psi\mapsto\ket\phi $ with
	\begin{equation}\label{eq:pure_state_prob_def}
        \begin{aligned}
            \ket\psi&=\frac{1}
            {\sqrt{6}}{\left(\ket{0,0}+\ket{1,0}+2\ket{\frac32,\frac32}\right)},\\ \ket\phi&=\ket{\frac12,-\frac12},
        \end{aligned}
	\end{equation}
	is possible with probability $ p=\frac{1}{3} $ \cite{supp}.
\end{ex}

\subsection{Transformations of mixed states} \label{subsec:tramix}
The methods presented above are well suited for \emph{pure} initial and target states. Here we present the generalization to mixed states: the basic idea is to express the quantum channel by the known form of its Kraus operators (from \cref{prop:kraussu2}) with proper constraints in order to parameterize the end state. Then, an optimization is performed to evaluate the metric of choice: here we use the \emph{maximum fidelity}. The basic problem statement is thus: for a given initial $\rho$ and target $\sigma$, maximize 
\begin{equation}\label{key}
	\mathcal{F}(\EE(\rho),\sigma)=\tr(\sqrt{\sqrt{\sigma}\EE(\rho)\sqrt{\sigma}})
\end{equation}

over all $\SUT$-covariant quantum channels $\EE$. Suppose for a while that the channel $\EE$ is fixed. Then, according to \cref{prop:fid}, the fidelity between $\sigma$ and $\EE(\rho)$ can be found as the solution to the semidefinite optimization problem of 

\begin{equation}\label{zeroth_opt_problem}
	\begin{aligned}
		\max_{X} \quad & \frac{1}{2}\tr(X+X^\dagger), \\
		\text{s.t.}\quad & 
		\begin{pmatrix}
			\sigma & X \\ X^\dagger & \EE(\rho)
		\end{pmatrix}\succcurlyeq 0.
	\end{aligned}
\end{equation}

The further optimization over $\EE$ naively might be understood as a complex nonlinear maximization over the coefficients appearing in \cref{eq:krausformsu2} with the constraints set by \cref{norm_channel}. It does, however, permit a semidefinite relaxation, which can be proven to return a strict solution. 

Using the nomenclature of \cref{eq:krausformsu2}, the output state $ \EE(\rho) $ contains quadratic expressions $ f_{J,\alpha}^{(j^\prime,j)}f_{J,\alpha}^{(k^\prime,k)\ast} $. The output can be linearized by defining the parameter matrices $ F_{J,\alpha}=\vec{f}_{J,\alpha}\vec{f}_{J,\alpha}^{\; \dagger} $
{, and considering optimization over full rank matrices $F_J=\sum_\alpha F_{J,\alpha}$.}
\begin{prop}\label{prop:fid_SDP}
Consider an initial mixed state $\rho$ and a target $\sigma$. The maximal fidelity between the output state $\EE(\rho)$ and the target $\sigma$ optimized over $\SUT$-covariant channel $\EE$ is the optimum of the following semidefinite optimization problem:   
\begin{equation}
	\begin{aligned}
		\max_{X,\{F_J\}} \quad & \frac{1}{2}\tr(X+X^\dagger), \\
		\text{s.t.}\quad & 
		\begin{pmatrix}
			\sigma & X \\ X^\dagger & \EE(\rho,\{F_J\})
		\end{pmatrix}\succcurlyeq 0, \\
		& \xi(\{F_J\})=0, \\
		& F_J\succcurlyeq 0, \;\forall J. \\
	\end{aligned}
\end{equation}
{Here, the linear constraints $\xi$ are the analogues of \cref{eq:normkraus} and \cref{norm_channel} expressed for the matrices $F_J$ instead of vectors $f_{J,\alpha}$.}
\end{prop}
\begin{proof}
    See \cref{app:proof_fid}.
\end{proof}

If the target state $ \sigma=\kb{\sigma}{\sigma} $ is pure, the semidefinite program can be simplified. The pure state fidelity is $ \mathcal{F}(\EE(\rho),\sigma)= \tr(\sqrt{\bra{\sigma}\EE(\rho)\ket{\sigma}})=\sqrt{\bra{\sigma}\EE(\rho)\ket{\sigma}} $ and  the squared fidelity $ \mathcal{F}^2(\EE(\rho),\sigma)=\bra{\sigma}\EE(\rho)\ket{\sigma}=\tr(\sigma\EE(\rho)) $ can be optimized instead.
\begin{prop}\label{prop:fid_SDP_pure}
For any mixed initial state $\rho$ and pure target state $\sigma$, the maximum attainable fidelity between $\sigma$ and $\EE(\rho)$ optimized over $\SUT$-covariant channels $\EE$ is the optimum of the following semidefinite optimization problem:
\begin{equation}
	\begin{aligned}
		\max_{\{F_J\}} \quad & \tr(\sigma\EE(\rho)), \\
		\text{s.t.}\quad & \xi(\{F_J\})=0, \\
		& F_J\succcurlyeq 0, \;\forall J. \\
	\end{aligned}
\end{equation}
\end{prop}
This observation can be implemented numerically -- as an illustration, let us consider the states defined in \cref{eq:pure_state_prob_def}.
\begin{ex}
	With the states defined as in \cref{eq:pure_state_prob_def} of \cref{ex:pure_state_prob}, the maximum achievable fidelity \cite{supp} via $\SUT$-covariant channels is
	\begin{equation}\label{key}
		\mathcal{F}(\EE_{opt}(\kb{\psi}{\psi}),\kb{\phi}{\phi})\approx 0.93~.
	\end{equation}
	This matches the earlier observation that the transformation $ \ket\psi\mapsto\ket\phi $ is not possible deterministically. 
\end{ex}
Here, we can also determine the fidelities for transformations where the total angular momentum is increased, whereas the characteristic function approach would always return a vanishing probability due to \cref{prop:max_j}: any spin increase is forbidden in the context of pure state transformations with postselection on the target state. {The fidelity can still be nonzero exactly because the optimal output state is not pure.}
\begin{ex}
\label{ex:maxfixfinalex}
	The maximum fidelity for a $\SUT$-covariant transformation $ \ket\psi=\ket{\frac32,\frac32}\mapsto\ket\phi=\ket{2,2} $ numerically converges to \cite{supp}
	\begin{equation}\label{key}
        \mathcal{F}\left(\EE_{opt}\left(\kb{\psi}{\psi}\right),\kb{\phi}{\phi}\right)=\sqrt{\frac45}\approx 0.89,
	\end{equation}
	with 
    \begin{equation}
     \EE_{opt}\left(\kb{\psi}{\psi}\right)={\frac45\kb{2,2}{2,2}+\frac15\kb{2,1}{2,1}}.
     \end{equation}
\end{ex}

\section{\label{sec:interferometer}Phase estimation experiment}
\begin{figure}
    \centering
    \includegraphics[width=\linewidth]{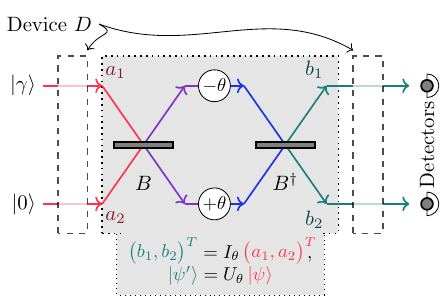}
    \caption{An unknown phase shift $\theta$ can be determined by measuring difference of the photon numbers at the output of the interferometer. The accuracy of phase estimation can be increased by the action of a device $D$ transforming the initial state $\ket\gamma\otimes\ket0$ into a more sensitive one $\ket\gamma\otimes\ket\tau$. Since \emph{the device action commutes with the interferometer}, the preparation stage can be replaced by postprocessing by the same device, with the same metrological improvement.}
    \label{fig:interf}
\end{figure}
Consider a phase estimation experiment pictured in \cref{fig:interf}. Here, two quantum light modes are transformed by an interferometric setup corresponding to an operation $\mathcal{U}_\theta$ and a device $D$ in two possible orders: $D \circ \mathcal{U}_\theta$ and $\mathcal{U}_\theta \circ D$. Subsequently, the photon numbers of the two output modes are measured. In the linear optics description, the action of a beamsplitter
\begin{equation}
    B=\frac{1}{\sqrt2}\begin{pmatrix}
    1 & i\\ i & 1
    \end{pmatrix}
\end{equation}
can be described as Heisenberg transformation of the annihilation operators:
\begin{equation-aligned}
    \begin{pmatrix}a_1\\a_2\end{pmatrix}\mapsto B\begin{pmatrix}a_1\\a_2\end{pmatrix},
\end{equation-aligned}
while the phase shift operation
\begin{equation}
    P_\theta=\begin{pmatrix}
        e^{i\theta} & 0 \\ 0 & e^{-i\theta}
    \end{pmatrix}
\end{equation}
corresponds to straightforward multiplication, $(a_1,a_2)\mapsto (e^{i\theta} a_1, e^{-i\theta} a_2)$. The overall transformation $I_\phi=B\cdot P_\theta \cdot B^\dagger$ leads to output modes
\begin{equation}
    \begin{pmatrix} b_1 \\ b_2 \end{pmatrix} = \overbrace{\begin{pmatrix}
        \cos\theta&-\sin\theta\\\sin\theta&\cos\theta
    \end{pmatrix}}^{I_\theta} \begin{pmatrix} a_1 \\ a_2 \end{pmatrix}
\end{equation}

In the Schr\"odinger view, the mode transformations correspond to a unitary operator transforming the initial state $\ket\psi$:
\begin{equation-aligned}
    \ket\psi \mapsto \underbrace{e^{\theta(a_2^\dagger a_1-a_1^\dagger a_2)}}_{U_\theta} \ket\psi.
\end{equation-aligned}
The interferometer action, described by a $\UO$ subgroup of $\SUT$ (which describes mode transformations more generally) formed by the matrices $I_\theta$, acts on quantum states via its {unitary representation $U_\theta$,  on the level of mixed states described by $\mathcal{U}_\theta(\rho) =U_\theta \rho U_\theta^\dagger$}. Interestingly, there exist operations $D$ \emph{commuting with all $\mathcal{U}_\theta$}: they can be placed before or after the interferometer, with the same overall operation of the interferometer and device in total. With such operations, \emph{the preparation step is equivalent to postprocessing}, and they can be made to perform useful work, as shown in the following example.

\begin{ex}
    Consider an interferometer with its input arms initialized in the coherent and vacuum state, respectively: $\ket\psi = \ket{\gamma}\ket{0}$. There exists a device $D$, for which $\mathcal{U}_\theta\circ D=D \circ \mathcal{U}_\theta$, transforming the state $\ket\psi$ with nonzero probability $p$ into
    \begin{equation}
\ket\phi=\ket{\gamma(1-\varepsilon)}\ket{\tau},
\label{eq:interfketphidef}
\end{equation}
where $\varepsilon>0$ and the $\ket\tau$ is a two-photon approximation of the squeezed vacuum,
\begin{equation-aligned}
    \ket\tau = \cos\tau\ket{0} -\sin\tau\ket{2}.
    \label{eq:interfkettaudef}
\end{equation-aligned}
\end{ex}
\begin{proof}
The possibility of state transformation is ascertained by the structure of the $\UO$ characteristic function associated with the interferometer action -- see \cref{corr:stochastic_trafo} and \cref{app:interf}. In this case, the characteristic function of a state $\ket\psi$ takes the form of
\begin{equation-aligned}
    \chi_\psi(\theta) = \langle\psi\vert U_\theta\vert\psi\rangle = \sum_{k\in \ZZ} C_k e^{i k \theta},
\end{equation-aligned}
where the $C_k$ have to be nonnegative and sum up to 1. The interconversion to a state $\ket\phi$ with the characteristic function of $\chi_\phi = \sum_{k\in\ZZ} e^{i k \theta} P_k$ is possible with probability $p$ if  $C_k=p P_k+(1-p)Q_k$, where $Q_k\ge 0$. This is possible in this case:
 the characteristic functions of $\ket\psi$ and $\ket\phi$ read (with $z:=\exp(i\theta)$)
\begin{equation-aligned}
    \chi_{\psi(\gamma)} &= \sum_{k\in\ZZ} z^k \overbrace{e^{-\abs{\gamma}^2}I_k(\abs{\gamma}^2)}^{C_k},\\
    \chi_\phi &= \chi_{\psi(\gamma(1-\varepsilon))} \times (A_{-4}z^{-4}+\ldots+A_4 z^4)\\
    &=\sum_{k\in\ZZ} P_k z^k,
\end{equation-aligned}
where $I_k$ is the modified Bessel function and the coefficients $A_l$ ($l=-4,\ldots,4$) are expressions involving $\gamma, \tau$ and $\varepsilon$. By analyzing the asymptotic behavior of $P_k$ and $C_k$, one can show that for any $\varepsilon>0$,
\begin{equation-aligned}
    p:=\min_{k\in\ZZ} \frac{C_k}{P_k}>0.
\end{equation-aligned}
The details of the calculations, including the definition of $A_l$ and proof that $p>0$, are presented in \cref{app:interf} and \cite{supp}.
\end{proof}

Importantly, the state $\ket\phi$ offers a metrological advantage over the coherent input $\ket\psi=\ket\gamma\ket0$: the vacuum state $\ket0$ is transformed to an approximation of the squeezed vacuum state $\ket\tau$. This is visible in the output of the entire system: with the difference of the output arm photon numbers defined as
\begin{equation-aligned}
    \delta N=b_1^\dagger b_1 - b_2^\dagger b_2,
    \label{eq:deltan}
\end{equation-aligned}
the phase estimation uncertainty $\Delta\theta$ is \cite{scully1997quantum}
\begin{equation}
    \Delta\theta=\sqrt{\Delta^2 (\delta N)}\abs{\dv{\expval{\delta N}}{\theta}}^{-1},
\end{equation}
where $\Delta^2 (\delta N)=\expval{\delta N^2}-\expval{\delta N}^2$ is the photon number variance. For both states $\ket\psi$ and $\ket\phi$, the maximum accuracy is achieved around $\theta\approx\frac\pi2$; in this case (for large $\gamma$ and optimal angle $\tau$ -- see \cref{app:interf}  for details)
\begin{equation-aligned}
    \Delta\theta_\psi &= \frac{1}{2\abs{\gamma}},\\
    \Delta\theta_\phi &\approx \Delta\theta_\psi \times \overbrace{\sqrt{3-\sqrt6}}^{\approx 0.74}.
\end{equation-aligned}

\section{\label{sec:summary} Conclusions}
The most important results of this work are the Propositions \labelcref{prop:poly_det} through \labelcref{prop:fid_SDP_pure}: they utilize a novel description of the $\SUT$ characteristic functions in terms of polynomial expressions in order to answer basic questions related to rotationally covariant state transformations. Such a characterization enables the direct use of the more general characteristic function theory, in a manner similar to the application of polynomials in the case of $\UO$. {The $\SUT$ transformations can be interpreted as passive linear optics mode mixing, and the found results are applied in \cref{sec:interferometer} to show the possiblity of state transformation in an inteferometric setting improving metrological sensitivity -- this transformation can be realized as either a state preparation or postprocessing step, since it explicitly commutes with the interferometer action.} 

Solutions to the $\SUT$-covariant state transformation problems presented in this work provide answers to the basic questions of quantum information science. The basic interactions found in the physical world do not depend on the frame of reference. This leads to the angular momentum being conserved, but further constraints on the allowed output states can be made by application of our observations. The same structure arises if a quantum state of angular momentum is prepared with respect to an unknown orientation: it can be deterministically transformed to only a subset of quantum states, characterized by \cref{prop:poly_det}.

{The results presented here are also applicable in the abstract theory of quantum reference frames, due to significant simplification of calculations involving $\SUT$ rotations through the polynomial parametrization (\cref{prop:su2_reps} and \cref{prop:charaspoly}). These observations together with known state transformation conditions offers a way to determine possible transformations between resource states. While it is known that deterministic transformations are unable to increase the relative orientation information, through \cref{prop:poly_det} it is possible to verify if a single reference $\ket\psi$ can be divided into a pair $\ket\phi$ and $\ket\xi$, and as exemplified in \cref{sec:interferometer}, probabilistic amplification is also possible in some cases.}

The methods presented in this paper require the target state to be known beforehand. We plan to generalize the procedures by developing a classification of the achievable states through decomposition of the initial characteristic function. Further research is also needed in order to understand the accessibility structure of metrologically important states of angular momentum (e.g., squeezed or Dicke states).

This work was supported by
the Deutsche Forschungsgemeinschaft  (DFG, German Research
Foundation, project numbers 447948357 and 440958198), the
Sino-German Center for Research Promotion (Project M-0294),
the ERC (Consolidator Grant 683107/TempoQ), the German
Ministry of Education and Research (Project QuKuK, BMBF Grant
No. 16KIS1618K), and the European Union’s Horizon 2020 research and
innovation programme under the Marie Skłodowska-Curie grant agreement
(No 945422). FO acknowledges support by the Studienförderfonds der Universität Siegen. We would like to thank David Jennings for the motivation to carry out this research.

\newpage
\bibliography{apssamp_arxiv}

\newpage
\appendix
\newpage
\onecolumngrid
\section{Majorana stars and polynomial $\SUT$ representation \label{app:majorana}}

\begin{figure}
    \centering
    \includegraphics[width=.5\linewidth]{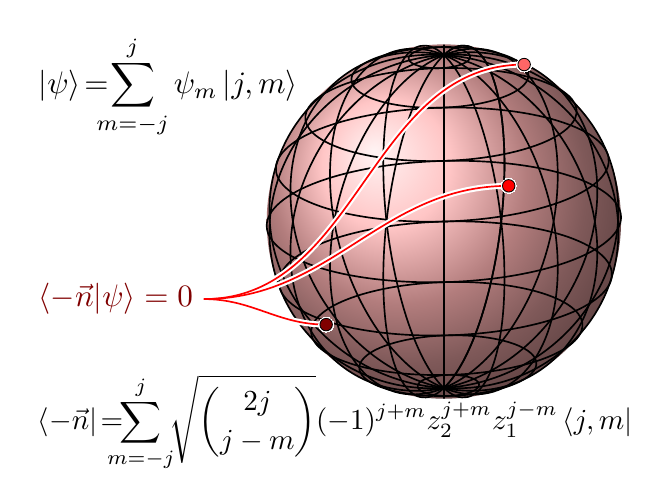}
    \caption{For a pure spin state $\ket\psi$ with definite total angular momentum $j$, the Majorana stars is a collection of $2j$ points on an unit sphere fully defining the state in question. The points correspond to directions $\vec n$ along which the state has zero overlap with a coherent state, $\langle -\vec n \vert \psi \rangle=0$.}
    \label{fig:majoranas}
\end{figure}
In this section, we aim to present the concept of Majorana stars and their connection to homogeneous polynomials. The goal is to provide an intuitive introduction to these mathematical entities, emphasizing their relevance in understanding spin-$j$ states and their relationship to the action of $\text{SU}(2)$ unitaries.
In addition to the explicit description through the amplitudes $\psi_m$, there exist other approaches: Majorana stars \cite{Majorana1932}  and homogeneous polynomials \cite{hall2003lie} are examples. The former are collections of points (also called \emph{constellations}) on a sphere, which can be regarded as roots of the associated polynomials.

Indeed, this is the traditional way to define the star representation \cite{Bengtsson2006}. It can be defined also through the use of \emph{spin coherent state} (\cref{def:coherentspin}). Such states are defined up to a phase only, but this does not pose an issue here: the Majorana stars are related to the probabilities $\lvert \langle \vec n \vert \psi \rangle \rvert^2$.  

This definition is sufficient to provide an informal definition of the collection of Majorana stars: intuitively, for a state of the form shown in \cref{eq:defjstate}, the 

the constellation is formed by the directions $\vec n$ for which the probability of finding an antipodal coherent state $\ket{-\vec n}$ is zero: 
\begin{equation-aligned}
    \langle -\vec n \vert \psi \rangle =0.
    \label{eq:majcohoverlap}
\end{equation-aligned}
The antipodal points are used in the definition exactly so that the coherent state $\ket{\vec n}$ has the Majorana constellation consisting only of the vector $\vec n$ itself: $\braket{-\vec n}{\vec n}=0$. To formalize this intuition (and take multiplicities into account), the polynomial parametrization of coherent states (\cref{eq:cohdefpoly}) can be used. Conveniently, if a coherent ket $\ket{\vec n}$ is described by a pair of complex numbers $(z_1, z_2)$, the antipodal ket $\ket{-\vec n}$ corresponds to $(-z_2^*,z_1^*)$ and subsequently
\begin{equation-aligned}
    \braket{-\vec n}{\psi} &= \left(\sum_{m=-j}^j \sqrt{\binom{2j}{j-m}}  (-z_2^*)^{j+m} (z_1^*)^{j-m} \ket{j,m}\right)^\dagger\left(\sum_{m=-j}^j \psi_m\ket{j,m}\right)\\
   &= 
    \!\! \sum_{m=-j}^{j}\!\sqrt{\binom{2j}{j-m}} (-1)^{j+m} z_2^{j+m}z_1^{j-m} \psi_{m}.
    \label{eq:majcohoverlap}
\end{equation-aligned}

The Majorana stars are exactly the directions $\vec n$ corresponding to the projective roots of this polynomial (taken with multiplicities) -- see \cref{fig:majoranas}:
\begin{defi}
    The Majorana constellation of a state $\ket\psi = \sum_{m=-j}^j \psi_m \ket{j,m}$ with definite spin $j$ is a collection of of $2j$ sphere points $\{\vec n^{(i)}\}_{i=1}^{2j}$ {such that \cref{eq:majcohoverlap} holds}. Each $\vec n^{(i)}$ in the collection corresponds to a point on the complex sphere $(z_1^{(i)}, z_2^{(i)})$ through \cref{eq:expvals}; the multiplicity of $\vec n$ is equal to the root multiplicity of this tuple in \cref{eq:majcohoverlap}.
    \label{def:majoranas}
\end{defi}
For instance, the constellation of a coherent state $\ket{\vec n}$ is $2j$ copies of $\vec n$. The star representation of any pure state is unique, provided it is fully contained within a single irreducible representation $j$. 

The Majorana representation provides an interpretation of the polynomial form (\cref{eq:majdefeq})); the key insight is that they demonstrate the action of \(\SUT\) unitaries via polynomial parametrization. The stars rotate rigidly in accordance with the $\SUT$ rotations: to show what is meant by that, pick a state $\ket\psi$ with the constellation $\{\vec n^{(i)}\}_{i=1}^{2j}$, and a $\SUT$ element which corresponds (in representation) to $U$ acting on $\ket\psi$ and $O\in\text{O}(3)$ being the related rotation on 3-dimensional real space. Then, the state $U \ket\psi$ has the constellation of  $\{O\vec n^{(i)}\}_{i=1}^{2j}$.

As previously mentioned, polynomial representation of $U^{(j)}$ can be expressed in closed form as a hypergeometric function. Taking terms independent  from $a$ in the sum within \cref{eq:rep_polysum}, one can obtain
\begin{equation-aligned}
    U^{(j)}_{m',m}=\sqrt{\frac{\binom{2 j}{j-m}}{\binom{2 j}{j-m'}}} v^{m-m'} \left(u^*\right)^{j-m}
   u^{j+m'} \sum_{\!\!a=\max \left\{0,m'-m\right\}}^{\!\!\min\left\{\!j\!-\!m\!,j\!+\!m'\!\right\}}  \binom{j-m}{a}  \binom{j+m}{a+m-m'} \left(-\frac{v v^*}{u u^*}\right)^a
\end{equation-aligned}
 The inside sum can be identified to be the hypergeometric $\,_2 F_1(n,l;k;x)$ for proper choice of the arguments: if one of $n, l$ is negative, it is polynomial in $x$ proportional to the above form (see  \cite{abramowitz1970}, sec. 15.4) and consequently

\begin{equation-aligned}
U^{(j)}_{m',m}=\binom{j+m}{m-m'}  \sqrt{\frac{\binom{2 j}{j-m}}{\binom{2 j}{j-m'}}} &\times v^{m-m'} \left(u^*\right)^{j-m}
   u^{j+m'} \\&\times\, _2F_1\left(m-j,-j-m';m-m'+1;-\frac{v v^*}{u u^*}\right).
\end{equation-aligned}
for $m\ge m'$. For $m<m'$, a pole in $\,_2 F_1$ appears; it is regularized by the binomial prefactor $\binom{j+m}{m-m'}$ written as ${\Gamma(j+m+1)}\left({\Gamma(m-m'+1)\Gamma(j+m'+1)}\right)^{-1}$ and 
\begin{equation-aligned}
U^{(j)}_{m',m}=\binom{j-m}{m'-m} \sqrt{\frac{\binom{2 j}{j-m}}{\binom{2 j}{j-m'}}}&\times (-1)^{m'-m} u^{j+m} \left(v^*\right)^{m'-m} \left(u^*\right)^{j-m'} \\
   &\,
   \times \,_2F_1\left(-j-m,m'-j;-m+m'+1;-\frac{v v^*}{u u^*}\right),
   \end{equation-aligned}
for $m< m'$.
\section{Strict treatment of the characteristic function in the polynomial form}\label{app:proof_max_j}
Let us define the maximal occupied total spin,
\begin{equation}\label{eq:maxj}
    \jmax{\rho}=\max \{j\;|\; \tr(\Pi_j\rho\Pi_j)>0\},
\end{equation}
where $\Pi_j$ is the projector onto total spin-$j$ subspace: $\Pi_j=\sum_\alpha\sum_{m=-j}^j \ket{j,m,\alpha}\bra{j,m,\alpha}$. First of the mentioned problems is that multiple polynomials correspond to the same characteristic function; however this problem admits a simple solution. Note that for any $\chi_\psi(u,u^*,v,v^*)$, the expression $\chi_\psi+(u u^* + v v^* -1) f(u,u^*,v,v^*)$ evaluates to the same value as $\chi_\psi$ on all group elements $(u,v)$ for an arbitrary polynomial $f$ as $u u^*+v v^*=1$ holds. This is exactly the constraint of $\SUT$ elements: to take care of the ambiguity, we use the properties of Gr\"obner bases with respect to polynomial division \cite{cox1994ideals} in order to define the \emph{canonical version} $\tilde \chi$ of a polynomial $\chi$:

\begin{defi}
    The \emph{canonical} polynomial  of a multivariate polynomial $\chi(u,v,u^*,v^*)$ is uniquely defined by
\begin{equation}
    \tilde \chi (u,v,u^*,v^*) = \chi - (u u^* + v v^* -1) f,
\end{equation}
where the polynomial $f$ is chosen such that the result of polynomial division of $\tilde \chi$ by $(u u^* + v v^* -1)$ has a zero quotient. 
\end{defi}
With that in mind, let us consider the linear operator $\pi$ first mentioned in \cref{def:coefflist}, taking as an input a polynomial in $u,v,u^*,v^*$, e.g.,
\begin{equation}
    \chi_\psi(u,u^\ast,v,v^\ast)=\sum_{\substack{a+b+c+d\leq\jmax{\psi} \\ a,b,c,d\geq 0}} \chi_{a,b,c,d}u^au^{\ast b}v^cv^{\ast d}.
\end{equation}
The output of $\pi$ is the list of coefficients of the canonical polynomial $\tilde \chi$, up to the total spin of $\jmax{\psi}$ (see \cref{eq:maxj}):
\begin{equation-aligned}
\pi[\chi_\psi]=(\tilde \chi_{abcd})_{a+b+c+d\leq\jmax{\psi}},
\end{equation-aligned}
such that
\begin{equation}
    \tilde\chi_\psi(u,u^\ast,v,v^\ast)=\sum_{\substack{a+b+c+d\leq\jmax{\psi} \\ a,b,c,d\geq 0}} \tilde \chi_{a,b,c,d}u^au^{\ast b}v^cv^{\ast d}.
\end{equation}

The coefficient list $\pi[\chi]$ is unique; and two polynomials can be compared without ambiguity for equality of the characteristic functions. Thus, the condition of unitary $\SUT$ interconvertibility (\cref{prop:directreach} applied to this case) of $\ket\psi$ and $\ket\phi$ takes the form of
\begin{equation}\label{key}
	\pi[{\chi}_\psi]=\pi[{\chi}_\phi].
\end{equation}
In order to apply the formalism presented in \cref{sec:theo_background}, the dimension of the Hilbert space has to be constrained. Starting from the equality of the characteristic functions
\begin{equation}\label{eq:prob_char}
    \chi_\psi (g)=p\chi_\xi(g)\chi_\phi(g)+(1-p)\chi_\sigma(g)
\end{equation}
for a probabilistic transformation $\ket\psi\mapsto\ket\phi$ required in \cref{corr:stochastic_trafo}, it is clear that the product state $\ket\xi\otimes\ket\phi$ and the state $\ket\sigma$ cannot have support on an irrep where $\ket\psi$ has vanishing support{: this is only possible with superposition, but the two states are effectively mixed together.}

Hence, it is possible to apply \cref{prop:iman1} and \cref{corr:stochastic_trafo} to the group $\SUT$ with all states having finite support and being contained in the Hilbert space
\begin{equation}
    \HH=\oplus_{j=0,\frac12,\dots,\jmax\psi}\operatorname{span}\left(\{\ket{j,m}\}_{m=-j}^j\right)
\end{equation}
as the only one-dimensional irrep of $\SUT$ is trivial and contained in $\HH$.

\begin{prop}\label{prop:max_j}
The maximal occupied angular momentum representation $\jmax{\xi}$ of the state $\ket\xi$ appearing in \cref{prop:iman1} can be constrained:
\begin{equation}
    \jmax{\xi} \leq \jmax{\psi}-\jmax{\phi}.
\end{equation}
Similarly, keeping with the nomenclature of \cref{corr:stochastic_trafo}, if $0<p<1$, then
	\begin{equation-aligned}\label{key}
		\jmax{\xi} &\leq \jmax{\psi} -\jmax{\phi},\\
  \quad \jmax{\sigma} &\leq \jmax{\psi}.
	\end{equation-aligned}
\end{prop}
\begin{proof}
	
The characteristic function of a state $ \ket\rho=\sum_{j\in\{0,\frac12,\dots,\jmax{rho}\}} \ket{\rho^{(j)}} $ is given by 
 \begin{equation-aligned}
     \chi_\rho=\sum_{j\in\{0,\frac12,\dots,\jmax{\rho}\}} \chi_{\rho^{(j)}},
 \end{equation-aligned}
 and all $ \chi_{\rho^{(j)}} $ are positive semidefinite as functions over $\SUT$. Recall that $ \chi_\xi\chi_\phi=\chi_{\phi\otimes\xi}$ and $ \ket\phi\otimes\ket\xi $ has the maximal angular momentum component of $ \jmax{\phi\otimes\xi}=\jmax{\phi}+\jmax{\xi}$ \footnote{$ \ket\phi\otimes\ket\xi $ is a product state and the Clebsch-Gordan coefficients of the form share $ C_{j_1,m_1;j_2,m_2}^{j_1+j_2,m_1+m_2} $ the same sign and thus cannot cancel out.}.

This must not be larger than $ \jmax{\psi} $, because the characteristic functions on different irreps are linearly independent. As {the part of $ \chi_{\phi\otimes\xi}$ corresponding to the irrep $j=(\jmax{\xi}+\jmax{\phi})$} is nonzero and positive semidefinite, we get
	\begin{equation}\label{key}
		\jmax{\xi}\leq \jmax{\psi}-\jmax{\phi}.
	\end{equation}
	The same holds for the characteristic function $ \chi_\sigma $ resulting in
	\begin{equation}\label{key}
		\jmax{\sigma}\leq\jmax{\psi}.
	\end{equation}
	This also implies that the maximum irrep in the support of a state cannot increase in probabilistic transformations $ \ket\psi\mapsto\ket\phi $.
\end{proof}

\section{Proof of \cref{prop:poly_det}}\label{app:proof_det}
By application of \cref{prop:iman1}, the existence of a $\SUT$-covariant channel realizing the transformation is equivalent to the existence of a particular state $\ket\xi$ such that the characteristic functions factorize: $\chi_\psi = \chi_\xi \chi_\phi$. The product corresponds to the characteristic function of the tensor state, $\chi_{\xi\otimes\phi}$, and hence the $\SUT$-covariant transformation is possible if and only if a state $\xi$ exists such that
\begin{equation}\label{eq:char_SU2}
    \chi_\psi(g)=\chi_{\xi\otimes\phi}(g)~~\forall g\in \SUT~.
\end{equation}
This is possible if and only if the canonical polynomials describing the characteristic functions have the same coefficients.

\section{Proof of \cref{prop:fid_SDP}}\label{app:proof_fid}
The output state $\EE(\rho,\{F_{J,\alpha}\}$ is linear in the matrices $ F_{J,\alpha} $, which must be positive semidefinite and $\rank(F_{J,\alpha})\leq 1$ for all $ \alpha $ and $ J\in\{0,\frac{1}{2},1,\dots\} $ by construction. Thus, the maximum fidelity attained over $\SUT$-covariant channels can be found as a solution to the following (not semidefinite) optimization problem:
\begin{equation}\label{first_opt_problem}
	\begin{aligned}
		\max_{X,\{F_{J,\alpha}\}} \quad & \frac{1}{2}\tr(X+X^\dagger), \\
		\text{s.t.}\quad & 
		\begin{pmatrix}
			\sigma & X \\ X^\dagger & \EE(\rho,\{F_{J,\alpha}\})
		\end{pmatrix}\succcurlyeq 0, \\
		& \xi(\{F_{J,\alpha}\})=0, \\
		& F_{J,\alpha}\succcurlyeq 0, \;\forall J,\alpha, \\
		& \rank(F_{J,\alpha})\leq 1, \;\forall J,\alpha, \\
	\end{aligned}
\end{equation}
where $ \xi(\{F_{J,\alpha}\})=0 $ is the linear normalization constraint (see \cref{prop:kraussu2} and \cite{Gour_2008} for details: {they are analogues of \cref{eq:normkraus} and \cref{norm_channel} expressed for the matrices $F_J$ instead of vectors $f_{J,\alpha}$ appearing in the original problem formulation}). To make this an explicitly semidefinite optimization problem, it has to be shown that 
\begin{enumerate*}[label=(\roman*)]
\item there are only finitely many relevant Kraus operators (corresponding to the matrices $F_{J,\alpha}$) and
  \item the optimization can be performed without the rank constraint.
\end{enumerate*}

The first problem can be solved by observing that for a given input $\rho$ and a target state $\sigma$ with finite support ($\jmax{\rho},\jmax{\sigma} < \infty$ {as defined in \cref{eq:maxj}}), the value of $J$ (the representation index of Kraus operator) can be bounded. Terms with $ \abs{J-j^\prime}>\jmax{\rho} $ or $ \abs{J-k^\prime}>\jmax{\rho} $ vanish (see equation \eqref{eq:krausformsu2}) and only terms with $ \abs{J-j^\prime}\leq \jmax{\rho}$ and $ \abs{J-k^\prime}\leq \jmax{\rho} $ have to be considered. The further constraints stem from the state $ \sigma $: only nonzero terms must fulfill $ j^\prime\leq \jmax{\sigma}$ or $ k^\prime\leq \jmax{\sigma} $. Combining both inequalities, we get
\begin{equation}\label{key}
	J\leq \jmax{\rho}+\jmax{\sigma}.
\end{equation}
Hence, the Kraus decomposition of $ \EE(\rho) $ contains only finitely many terms and the search space for the matrices $F_{J,\alpha}$ is finite-dimensional.

In the original form of \cref{first_opt_problem}, only rank-$1$ matrices are taken into account. The more general case (of unrestricted rank optimization) contains this set, but may in principle be \emph{too} general: the maximization result could be just an upper bound. However, the unconstrained optimization is equivalent to the restricted case: any positive semidefinite matrix $F_J$ can be decomposed into rank-$1$ matrices $F_{J,\alpha}$ by
\begin{equation}\label{eq:fjdecomp}
	F_J=\sum_\alpha F_{J,\alpha}.
\end{equation}
Here, $F_{J,\alpha}$ are essentially rescaled projectors onto the one-dimensional eigenspaces of $F_J$. This translates directly to the description of the channel $\EE$ via the Kraus operators by linearity:
\begin{equation}\label{key}
	\EE(\rho,\{F_{J,\alpha}\})=\EE(\rho,\{F_{J}\}).
\end{equation} 
The normalization conditions (\cref{norm_channel}) hold automatically.\\

\section{Exponential parametrization examples \label{app:exppar}}
With unitary representation in spin $j$ defined as $U^{(j)} = {\exp(i (n_x J^{(j)}_x+n_y J^{(j)}_y+n_z J^{(j)}_z))}$ and 
\begin{equation-aligned}
    r&:=\sqrt{n_x^2+n_y^2+n_z^2},    &\nu:=n_x+i n_y,~\mu&:=n_x^2+n_y^2,\\
    c&:=\cos r,~s:=\sin r,    &c':=\cos \left(\frac{r}2\right),~s'&:=\sin \left(\frac{r}2\right),
\end{equation-aligned}
the following matrix corresponds to the $j=1$ representation:
\begin{equation-aligned}
    U^{(j=1)}=\frac{1}{4r^2}\left(
\begin{array}{ccc}
 2 (c+1) \mu +4 c n_z^2+4 i r s n_z & 4 i \sqrt{2} \nu ^* s' \left(r c'+i n_z s'\right) & 2 (c-1) \left(\nu ^*\right)^2 \\
 2 \sqrt{2} \nu  \left((c-1) n_z+i r s\right) & 4 \left(c \mu +n_z^2\right) & 2 \sqrt{2} \nu ^* \left(-(c-1) n_z+i r s\right) \\
 2 (c-1) \nu ^2 & 2 \sqrt{2} \nu  \left(-(c-1) n_z+i r s\right) & 2 (c+1) \mu +4 c n_z^2-4 i r s n_z \\
\end{array}
\right).
\end{equation-aligned}
Evidently, the polynomial parametrization found in \cref{eq:polyrepsu2} is better suited for symbolic calculations.
\section{Inferferometric experiment calculations \label{app:interf}}
The interferometer action is described by the following mode transformation:
\begin{equation}
    \begin{pmatrix} b_1 \\ b_2 \end{pmatrix} = \overbrace{\begin{pmatrix}
        \cos\theta&-\sin\theta\\\sin\theta&\cos\theta
    \end{pmatrix}}^{I_\theta} \begin{pmatrix} a_1 \\ a_2 \end{pmatrix}
    \label{eq:u1subu2form}
\end{equation}
The set of matrices $I_\theta$ forms a one-dimensional abelian group: $I_\theta I_\lambda = I_{\theta+\lambda}$; it is the $\UO$ subgroup of the set $\SUT$ of all unitary and unit determinant matrices of size $2$. The unitary action $U_\theta$ on quantum states of light:
\begin{equation-aligned}
    \ket\psi = \sum_{m,n\in\NN} \psi_{mn} \ket{m,n},
\end{equation-aligned}
is defined by the unitary representation $U_\theta$:
\begin{equation-aligned}
    \ket\psi\mapsto \underbrace{\exp(\theta(a_2^\dagger a_1-a_1^\dagger a_2))}_{U_\theta}\ket{\psi}.
\end{equation-aligned} 
Note that $I_\theta$ is still a subgroup of $\SUT$, and the result of \cref{prop:su2_reps} still holds here; the relevant representations here are supported on fixed total photon number subspaces spanned by $(\{\ket{n,0},\ldots,\ket{1,n-1},\ket{0,n}\})$, with the effective $j$ index equal to $j:=\frac{n}2$. In this case, the parameters can be inferred from the form of $\SUT$ characterication compared to \cref{eq:u1subu2form}: with $z:=\exp(i\theta)$,
\begin{equation-aligned}
    u=\overbrace{\frac{z+z^{-1}}{2}}^{\cos\theta},~v=\overbrace{\frac{z-z^{-1}}{2i}}^{\sin\theta}.
    \label{eq:su2tou1uv}
\end{equation-aligned}
The characteristic function of this subgroup can be thought of as the expectation value of the unitary,
\begin{equation-aligned}
    \chi_\psi(\theta) = \langle\psi\vert U_\theta \vert \psi\rangle,
\end{equation-aligned}
and under certain regularity conditions can be expanded (by noticing that the terms in $U_\theta$ are formed by positive and negative integer powers of $z$):
\begin{equation-aligned}
    \chi_\psi = \sum_{k\in\ZZ} C_k z^k, 
    \label{eq:chipsiu1decomp}
\end{equation-aligned}
with $z:=\exp(i\theta)$. An arbitrary function $f$ with similar decomposition $f:=\sum_{k\in\ZZ} C_k z^k$ is a characteristic function of some state $\ket\psi$ ($f=\chi_\psi$) if and only if all $C_k\ge0$ and $\sum_{k\in\ZZ}C_k=1$: this is evident by considering the fact that the irreducible representations of $\UO$ are one-dimensional{. Therefore, each one corresponds to a vector $\vert\!\vert{k;\alpha}\rangle\!\rangle$} indexed by $k\in\ZZ$ and multiplicity index $\alpha$; in such case, for the state
\begin{equation-aligned}
    \ket\psi=\sum_{k,\alpha} \psi_{k,\alpha} \vert\!\vert{k;\alpha}\rangle\!\rangle,
\end{equation-aligned}
the characteristic function has the form $\chi_\psi=\sum_{k\in\ZZ}(\sum_\alpha \lvert\psi_{k;\alpha}\rvert^2)z^k$, and the coefficients are explicitly nonnegative.

For the input consisting of coherent state $\ket\gamma$ in the mode $a_1$ and vacuum $\ket0$ in the mode $a_2$, the state is $\ket\psi =\sum_{n\in\NN} \exp{-\lvert\gamma\rvert^2}\frac{\gamma^n}{\sqrt{n!}} \ket{n,0}$. As is evident from the form of $\SUT$ parametrization (\cref{prop:su2_reps}) together with \cref{eq:su2tou1uv} (see also \cite{supp}), the elementary matrix element building the characteristic function is $\langle n',0\vert U_\theta \vert n,0\rangle=\delta_{n',n} u^n$, and subsequently

\begin{equation-aligned}
    \chi_\psi &= \sum_{n\in\NN}  \exp{-\abs{\gamma}^2}\frac{\abs{\gamma}^{2n}}{n!} \left(\frac{z+z^{-1}}{2}\right)^n \\
    &=\exp{-\abs{\gamma}^2} \exp\left(\abs{\gamma}^2 \left(\frac{z+z^{-1}}{2}\right)\right).
\end{equation-aligned}

This expression can be identified as the (rescaled) generating function of the modified Bessel function $I_k(\lambda)$ (see \cite{abramowitz1970}, section 9.6), and expanded to
\begin{equation-aligned}
\chi_\psi = \sum_{k\in\ZZ} \overbrace{\exp(-\abs{\gamma}^2) I_k(\abs{\gamma}^2)}^{C_k} z^k.
\end{equation-aligned}

The calculation of the characteristic function of $\ket\phi=\ket\gamma\ket\tau$, defined as in \cref{eq:interfketphidef} and \cref{eq:interfkettaudef}, follows a similar pattern. Then, the relevant matrix elements
{are}
\begin{equation-aligned}
    \langle n,0\vert U_\theta \vert m,0\rangle =&~ \delta_{m,n} u^n,\\
    \langle n\!-\!2,2\vert U_\theta \vert m\!-\!2,2\rangle =&~ \delta_{m,n} u^{n-4} \!\times\! \Big[\!-\!(2 n\!-\!4)\abs{u v}^2\!+\!\left(\frac{n}{2}\!-\!1\right)\!(n\!-\!3) \abs{v}^4\!+\!\abs{u}^4\Big],\\
    \langle n,0 \vert U_\theta \vert m-2,2\rangle =&~ \delta_{m,n} \sqrt{\frac{(n-1) n}{2}} \left(v^*\right)^2 u^{n-2},\\
    \langle n-2,2 \vert U_\theta \vert m,0\rangle =&~ \delta_{m,n} \sqrt{\frac{(n-1) n}{2}} v^2 u^{n-2}.
\end{equation-aligned}
With \cref{eq:su2tou1uv} and the expanded form of the state $\ket\phi$,
\begin{equation-aligned}
    \ket\phi =& \exp(-\abs{\gamma}^2) \Big[\left(\ket{0,0}+\gamma \ket{1,0}\right)\cos\tau \\&+ \sum_{n\ge 2} \alpha^n\left(\frac{\cos\tau\ket{n,0}}{\sqrt{n!}}-\frac{\sin\tau\ket{n-2,2}}{\alpha^2\sqrt{(n-2)!}}\right)\Big],
\end{equation-aligned}
the characteristic function can be found similarly to the case of $\ket\psi=\ket\gamma\ket0$:
\begin{equation-aligned}
    \chi_\phi = \chi_\psi \times \sum_{k=-4}^4 z^k A_k,
    \label{eq:chiphifromchipsi}
\end{equation-aligned}
where (see \cite{supp} for the calculation performed in Wolfram Mathematica):
\begin{equation-aligned}
    A_{\pm 4}=&\frac{1}{32} | \gamma | ^4 \sin ^2(\tau ),\\
    A_{\pm 3}=&\frac{1}{4} \abs{\gamma}^2 \sin ^2(\tau )\\
    A_{\pm 2}=&\frac{1}{16} \Big[-2 \left(| \gamma | ^4-2\right) \sin ^2(\tau )+\\
    &(\gamma ^2+\left(\gamma ^*\right)^2 ) \sqrt{1-\cos (4 \tau )}\Big],\\
    A_{\pm 1}=&-\frac{1}{4} \abs{\gamma}^2 \sin ^2(\tau ),\\
    A_0 =&1-2(A_1+A_2+A_3+A_4).
\end{equation-aligned}
By \cref{eq:chiphifromchipsi}, the coefficients $P_k$ in $\chi_\phi = \sum_{k\in\ZZ} P_k z^k$ are convolution of $C_k$ (of the characteristic function $\chi_\psi$) and $A_k$. 

Let us now determine reachability of $\ket\phi=\ket{\gamma(1-\varepsilon)}\ket{\tau}$ from $\ket\psi=\ket\gamma\ket0$. By \cref{corr:stochastic_trafo} (with auxiliary singlet state $\ket\xi$, such that $\chi_\xi=1$) applied to the decompositions of $\chi_\psi$, $\chi_\phi$ and $\chi_\sigma$ into Laurent series:
\begin{equation-aligned}
    \chi_\psi =& \sum_{k\in\ZZ} C_k z^k,\\
    \chi_\phi =& \sum_{k\in\ZZ} P_k z^k,\\
    \chi_\sigma =& \sum_{k\in\ZZ} Q_k z^k,\\
\end{equation-aligned}
with the restriction that the characteristic functions actually describe quantum states --  all $C_k, P_k, Q_k\ge0$ -- we conclude that the maximum probability of interconversion is such that for some $k'$, $C_{k'}= p P_{k'}$, and $Q_{k'}=0$ meaning that
\begin{equation-aligned}
    p=\inf_{k\in\ZZ} \frac{C_k}{P_k}.
    \label{eq:maxpckpk}
\end{equation-aligned}

For $C_k$ corresponding to the coherent input $\ket\gamma\ket0$ and $P_k$ describing the coefficients of the characteristic function of $\ket{\gamma(1-\varepsilon)}\ket\tau$, the support of both series are entire $\ZZ$, and the question whether $p$ defined in \cref{eq:maxpckpk} is nonzero is determined by the asymptotic properties of $C_k$ and $P_k$. 

For $\abs{k}\rightarrow \infty$, the modified Bessel function $I_k$ has the approximation of (see \cite{abramowitz1970}, section 9.6) 
\begin{equation-aligned}
    I_k(x) \approx \left(\frac{{x}}2\right)^{\abs{k}} \frac{1}{(\abs{k})!}.
\end{equation-aligned}
Upon insertion of the asymptotic form to $C_k$ and $P_k$, comparison of the dominant terms proves that extraction with positive nonzero $p$ is possible for arbitrary $\varepsilon>0$.

The variance of $\delta N$ defined in \cref{eq:deltan} for the state $\ket\gamma\ket\tau$ can be calculated by substituting $b$ and $b^\dagger $ in \cref{eq:deltan} with \cref{eq:u1subu2form}, and explicit determination of the resulting polynomial in $a_1, a_2, a_1^\dagger, a_2^\dagger$. In the calculation, standard commutation relation rules are used ($a_i a_j^\dagger = a_j^\dagger a_i + \delta_{i,j} $), together with the definitional property of coherent states ($a_1 \ket\gamma =\gamma\ket\gamma$) and explicit matrix calculations involving $a_2, a_2^\dagger$ and $\ket\tau$; the result is
\begin{equation-aligned}
    \!\!\Delta^2(\delta N)\!=& -(\gamma^2+(\gamma^*)^2)\frac{\sin ^2(2 \theta ) \sin (2 \tau )}{\sqrt{2}}\\
    &\!-\!2\abs{\gamma}^2\!\left(\! \sin
   ^2(2 \theta ) \cos (2 \tau )\!+\!\frac{\!\cos (4 \theta )\!}{2}\!-\!1\right)\\
   &+2 \sin ^2(\tau ) \left(\cos ^2(2 \theta ) \cos (2 \tau )+1\right),
\end{equation-aligned}
{and the expectation value has a simple form of}
\begin{equation-aligned}
    \langle{\delta N\rangle}=\cos (2 \theta ) \left(1-\abs{\gamma}^2-\cos (2 \tau )\right).
\end{equation-aligned}
For any given $\gamma$ and $\tau$, the maximal accuracy of phase determination, as measured by
\begin{equation-aligned}
    \Delta\theta = \frac{\sqrt{\Delta^2(\delta N)}}{\lvert \frac{\partial }{\partial \theta} \langle \delta N\rangle\rvert},
\end{equation-aligned}
is achieved for $\theta=\frac\pi4$. With this assumption,
\begin{equation-aligned}
    \Delta\theta=&\frac{\!\!\!\sqrt{3 \abs{\gamma}^2\!\!\!+\!1\!-\!\sqrt2 \Re(\gamma^2) {\sin (2\tau )}  \!-\!\left(2\abs{\gamma}^2\!\!+\!1\right)\!\cos (2 \tau )}}{2 \left(| \gamma | ^2+\cos (2 \tau )-1\right)}.
\end{equation-aligned}
with $\tau=0$, this reduces to the case of $\ket\psi=\ket\gamma\ket0$: $\Delta\theta=(2\abs{\gamma})^{-1}$. Since the only dependence of $\gamma$ phase is through $\Re\left(\gamma^2\right)$, which has to be as large as possible for smallest $\Delta\theta$, let us assume that $\gamma>0$ and then
\begin{equation-aligned}
    \Delta\theta =\frac{1}{2\gamma} \frac{\sqrt{(2 \sin ^2(\tau )){\gamma ^{-2}}-\sqrt{2} \sin (2 \tau )-2 \cos (2 \tau )+3}}{1+(\cos(2\tau)-1)\gamma^{-2}},
\end{equation-aligned}
which asymptotically for $\gamma\rightarrow\infty$ is minimal for $\theta = \arctan \sqrt{5-2\sqrt 6}$ and 
\begin{equation-aligned}
    \Delta\theta \approx \frac{1}{2\gamma}\times{\sqrt{3-\sqrt6}}.
\end{equation-aligned}
\end{document}